\documentclass[11pt]{article}
\usepackage[secthm]{evan}
\usepackage{url}
\usepackage{bbm}

\def\cynthia#1{{ {\color{red}{\bf CD:} #1}}}
\def\Elbert#1{{ {\color{blue}{\bf ED:} #1}}}
\newcommand{\Var}{\operatorname{Var}}
\newcommand{\E}{\mathbb{E}}

\newcommand{\Unif}{\operatorname{Unif}}
\usepackage{algorithm}
\usepackage{algpseudocode}

\newcommand{\Lap}{\operatorname{Lap}}
\usepackage{float}
\newcommand{\Interact}{\operatorname{Interact}}
\usepackage{xspace}

\newcommand{\Poly}{\operatorname{Poly}}

\newcommand{\iid}{\overset{\textit{iid}}{\sim}}
\newcommand{\pair}{\widetilde{C}^{\mathrm{pair}}}
\newcommand{\triple}{\widetilde{C}^{\mathrm{triple}}}

\newcommand{\Bigabs}[1]{\Bigl\lvert#1\Bigr\rvert}
\newcommand{\tO}{\widetilde{O}}
\newcommand{\tS}{\widetilde{S}}
\definecolor{navy}{rgb}{0, 0, 0.75}
\newcommand{\highlight}{\textcolor{navy}}
\newcommand{\Z}{\mathbb{Z}}

\newcommand{\Bin}{\mathrm{Bin}}

\DeclareMathOperator{\TV}{TV}

\newcommand{\cX}{\mathcal{X}}

\newcommand{\botel}{\bot}

\renewcommand{\eps}{\varepsilon}

\usepackage{csquotes}
\MakeOuterQuote{"}

\newcommand\blfootnote[1]{
    \begingroup \renewcommand\thefootnote{}\footnote{#1}
    \addtocounter{footnote}{-1}
    \endgroup
}

\title{Differentially Private Verification of Distribution Properties}
\author{
    Elbert Du \\
    Harvard University \\
    \url{edu1@g.harvard.edu}
\and
    Cynthia Dwork \\
    Harvard University \\
    \url{dwork@seas.harvard.edu}
\and
    Pranay Tankala \\
    Harvard University \\
    \url{pranay_tankala@g.harvard.edu}
\and
    Linjun Zhang \\
    Rutgers University \\
    \url{linjun.zhang@rutgers.edu}
}
\date{\today}

\begin{document}
\maketitle

\thispagestyle{empty}

\begin{abstract}

A recent line of work initiated by Chiesa and Gur and further developed by Herman and Rothblum investigates the sample and communication complexity of verifying properties of distributions with the assistance of a powerful, knowledgeable, but untrusted prover. 
In this work, we initiate the study of {\em differentially private} distribution property verification. After all, if we do not trust the prover to help us with verification, why should we trust it with our sensitive sample?
We map a landscape of differentially private prover-aided proofs of properties of distributions.

In the non-private case it is known that one-round (two message) private-coin protocols can have substantially lower complexity than public-coin (AM) protocols. In contrast, the possibility for improvement in differentially private interactive proofs  depends on the privacy parameter regime and the privacy model.  Drawing on connections between privacy and replicability and privacy amplification techniques in the literature we show for sample size $s$:
\begin{itemize}
\item There exists a reduction from {\em any} one-round $(\varepsilon,\delta)$-differentially private private-coin interactive proof to a one-round
    public-coin differentially private interactive proof with essentially the same privacy parameters, for the parameter regime $\varepsilon = O(1/\sqrt{s})$ and $\delta= O(1/s^{5/2})$, and with essentially the same sample and communication complexities. 
    \item If the verifier's message in the private-coin interactive proof is $O(1/\sqrt{\log s})$  {\em locally} differentially private -- a far more relaxed privacy parameter regime in a different model -- then applying one additional transformation again yields a one-round public-coin protocol with essentially the same privacy bound, sample complexity, and communication complexity.  

\item However, when the privacy guarantee is very relaxed  ($\varepsilon \in \Omega(\log s)$), private coins indeed reduce sample and communication complexities. 

\end{itemize}
We also obtain a computationally efficient Merlin-Arthur (one-message) proof for privately testing whether samples are drawn from a product distribution, and prove that its sample complexity is optimal in its dependence on $N$ and $\varepsilon$ (and within a $polylog N$ factor in its dependence on distance) by reducing uniformity testing to independence testing with Boolean attributes and appealing to known lower bounds on sample complexity for (even possibly inefficient) private uniformity testing.  
\end{abstract}

\blfootnote{This work was supported in part by Simons Foundation Grant 733782, Cooperative Agreement CB20ADR0160001 with the United States Census Bureau, NSF CAREER DMS-2340241 and Renaissance Philanthropy "AI for Math" Fund.

An earlier version of this paper contained an error; we withdraw the claims on private interactive arguments of proximity.  See Appendix~\ref{sec:errata}.}

\newpage
\thispagestyle{empty}
\tableofcontents
\newpage
\pagestyle{plain}
\setcounter{page}{1}

\section{Introduction}
Given sample access to a distribution $\mathcal{D}$, what can we learn about $\mathcal{D}$ with the aid of a powerful, but untrusted prover who has full access to $\mathcal{D}$? The study of proof systems for testing properties of {\em distributions} was initiated by Chiesa and Gur~\cite{chiesa_et_al:LIPIcs.ITCS.2018.53}, who obtained separations in sample complexity between the standard (no-prover) model and single-message (MA and NP) proof systems, as well as between MA and private-coin interactive proofs.  The subject was further developed by Herman and Rothblum in a sequence of  works on verifying quantities of a distribution~\cite{herman2022verifying,Herman23_doubly_efficient,herman2024verifyreasonabledistributionproperty,herman2024public}. For example, the prover may claim that $\mathcal{D}$ has entropy $m$, or support size $s$; the goal of the interaction is to verify that the claim is approximately correct.  

Our work was inspired by the remarkable finding that, for a distribution over $[N]$, while approximating certain basic quantities, including the support size, entropy, and distance to the uniform distribution, all require $\tilde{\Omega}(N)$ samples~\cite{NIPS2013_53c04118}, there are 2-message private-coin interactive proofs for verification of each of these quantities in which the verifier's sample complexity is only $\tilde{O}(\sqrt N)$~\cite{herman2022verifying}.\footnote{Communication and computation complexity are also important, but for clarity of exposition in this Introduction we restrict our attention to sample complexity.}  More generally, the proof system in~\cite{herman2022verifying} applies to any {\em label-invariant} (aka {\em symmetric}) property, in which membership in the property is closed under permutation of the elements in the domain.  

In this paper, we introduce a new dimension to this investigation:  If we cannot trust the prover to be honest in its claims about the distribution, why should we trust it with sensitive information in our dataset? 
Accordingly, we initiate the study of {\em differentially private}~\cite{dwork2006calibrating,Dwork_McSherry_Nissim_Smith_2017} interactive proofs of properties of distributions:

\emph{What can we learn from an untrusted prover while protecting the privacy of our sample?}

There are two places in which the verifier might leak information about its sample: in its messages to the prover, if any, and in its decision whether to accept or reject.
In the no-prover (standard) case, there are no messages, so we only need to worry about privacy leakage in the decision.  For this case, \cite{aliakbarpour2017differentiallyprivateidentitycloseness, aliakbarpour2017differentiallyprivateidentityclosenessArxiv} obtained property-testing algorithms whose decisions ensure differential privacy via an elegant generic reduction from non-private to private testers that increases sample complexity by at most a factor of $(1/\varepsilon)$, where $\varepsilon$ is the privacy parameter\footnote{This reduction appears in the full version of the paper, or \cite{aliakbarpour2017differentiallyprivateidentityclosenessArxiv}.}.  This generic transformation is easily extended to the case of public-coin (aka Arthur-Merlin or AM) proofs (Section~\ref{sec:generic}).  This is not surprising: in AM proof systems the verifier's messages are restricted to be random strings and are therefore independent of the sample, so, as in the no-prover case, the only threat to privacy arises in the decision. 

In~\cite{herman2024public}, Herman obtained an AM protocol for verifying any label-invariant property using $\tilde{\Omega}(N^{2/3})$ samples.  Unlike the protocols in the private-coin setting, this protocol does not match the best known lower bound in \cite{chiesa_et_al:LIPIcs.ITCS.2018.53}. This gap is currently unresolved.  We show that, by relaxing to approximate differential privacy~\cite{dwork2006our}, also known as $(\varepsilon,\delta)$-differential privacy, we can use the Propose-Test-Release framework of Dwork and Lei~\cite{dwork2009differential} to adapt Herman's non-private AM proof system  while paying a smaller price 
for privacy than that of the generic transformation (Section~\ref{sec:label_invariant}). The rough intuition is that, due to the randomness of the sample, even the non-private proof system has to cope with variability in certain counts of elements, and this uncertainty ``typically'' dominates noise introduced for privacy. Following the propose-test-release paradigm, we use a differentially private test to determine whether the situation is indeed``typical.''

The landscape for differentially private {\em private-coin} interactive proofs is more complicated.  First, we show that for small privacy losses, specifically, for the parameter regime $\varepsilon = O(1/\sqrt{s})$ and $\delta= O(1/s^{5/2})$, private coins buy us nothing, at least for one-round (2-message) protocols: In this regime, there exists a reduction from {\em any} one-round $(\varepsilon,\delta)$-differentially private private-coin interactive proof to a one-round public-coin differentially private interactive proof with only a constant factor increase in its privacy parameters, sample complexity, and communication complexity.  The proof of this result relies on a beautiful conversion from differential privacy to {\em replicability}, due to Bun, Gaboardi, Hopkins, Impagliazzo, Lei, Pitassi, Sivakumar, and Sorrel~\cite{bun2023stabilitystableconnectionsreplicability} (see also~\cite{reproducibility}). Speaking informally, a replicable algorithm operating on a sample is ``very likely'' to produce the identical output when run on a different sample from the same distribution. Suppose the verifier chooses its message using a replicable algorithm. Then, by replicability, the prover could instead use its knowledge of the distribution to generate a message that the verifier was likely to send, obviating the need for the verifier's message. We do not know if it is possible to completely do away with the verifier-to-prover message without increasing the sample complexity\footnote{Even in the non-private case there is no known separation between Merlin-Arthur proofs and Arthur-Merlin proofs for properties of distributions; see the discussion in Section 1.2 of~\cite{chiesa_et_al:LIPIcs.ITCS.2018.53}}. Starting from a differentially private algorithm and applying the transformation in~\cite{bun2023stabilitystableconnectionsreplicability} to get a replicable algorithm, we can replace the verifier's message with a random string which steers the prover towards certain messages that the verifier could have sent, yielding an AM protocol. The additional communication cost incurred is merely the length of this random string. Details for this and the next result appear in Section~\ref{sec:private-to-public}.

In {\em local} differential privacy, respondents randomize their data before sending them to the analyst. This model is popular in industry as it moves the trust boundary to the client, and indeed Apple, uses local differential privacy, for example, to privately compute statistics about emoji usage on iPhones \cite{apple2017local}. An interactive proof system uses local differential privacy if each element in the sample is treated using local differential privacy.  If the verifier's message in the private-coin interactive proof is $O(1/\sqrt{\log s})$ locally differentially private, then applying in addition a powerful result on privacy amplification via shuffling (a procedure in which the source of a message is disconnected from its contents) due to Feldman, McMillan, and Talwar~\cite{feldman2021hidingclonessimplenearly} yields a new reduction to a one-round public-coin protocol.  The new reduction again preserves the privacy bound and the sample and communication complexities of the private-coin protocol up to constant factors. This is a far more relaxed privacy parameter regime than for our first reduction, but the local differentially private model is more constrained than the standard, centralized, model.
    
In Section~\ref{sec:independence}, we derive what to our knowledge is the first (private or non-private) single-message (MA) proof protocol for testing whether a distribution is a product distribution.
When $s < N$, which is the parameter regime of interest to us, we prove this bound is tight in its dependence on the $N,\varepsilon$, and, up to $polylog N$ factors, in the distance parameter $\sigma$.  We do this by reducing uniformity testing to independence testing with Boolean-valued attributes, and appealing to a lower bound on the sample complexity of private uniformity testing due to~\cite{ASZ18}.
For the upper bound we derive the first {\em  private identity tester}, meaning that the tester's work is $O(s)$ which can be much less than $N$.
        

\subsection{Additional Related Work}
Distribution testing in the standard (no prover) model was introduced by Goldreich and Ron~\cite{goldreich2011testing} and Batu {\it et al.}~\cite{batu2000testing}.
There is a broad literature on this topic; see~\cite{canonne2020survey} for an excellent survey.  Closest to our work are results on identity and uniformity testing~\cite{goldreich2011testing,Paninski,diakonikolas2016collisionbasedtestersoptimaluniformity,DK16,diakonikolas2016collisionbasedtestersoptimaluniformity,ADK15}, Goldreich's reduction from closeness to uniformity~\cite{goldreich2020uniform}, which we review in Section~\ref{sec:identity-to-uniformity}, and $k$-wise or general independence~\cite{Alonkwise,ADK15}.

There is also a large literature on differentially private distribution testing.  Beyond the above-mentioned work of Aliakbarpour {\it et al.}~\cite{aliakbarpour2017differentiallyprivateidentitycloseness,aliakbarpour2017differentiallyprivateidentitycloseness}, we note the excellent work on tight sample complexity bounds of Acharya {\it et al.}~\cite{ASZ18}, whose lower bounds (in the absence of a prover) provide a natural baseline for evaluating our positive results when a prover is present.  
Canonne {\it et al.}~\cite{CKMUZ20} obtain differentially private identity testers for high-dimensional distribution families, including known-covariance Gaussians and product distributions over $\{\pm1\}^d$,   obtaining sample complexities that match the non-private minimax rate in several parameter regimes.  Subsequent improvements were obtained by Narayanan~\cite{Narayanan22}.
Differentially private distribution testing in the local model, in which each of $s$ parties holds a sample from the distribution, is studied in~\cite{pmlr-v89-acharya19b}.
%

Moving toward the interactive setting in the sense that the tester is provided with additional but not necessarily trusted information, \cite{AAS26} studies independence testing in a model in which the tester receives potentially untrusted predictive distributional information and remains valid in the worst case while improving its sample complexity when the prediction is accurate.  This work extends a line of investigation in augmented algorithms dating back to at least~\cite{canonne2014testing}.
%

We study the private verification of distribution properties via interactive proofs. Interactive proofs were introduced by Babai~\cite{babai1985trading} and Goldwasser, Micali, and Rackoff~\cite{goldwasser1985knowledge,goldwasser89} in the context of proving computational statements about an input that is fully known to the prover and fully accessible to the verifier. In distribution testing, the distribution is akin to the input, but the verifier only has access to it through sampling 

As noted earlier, the study of prover-assisted property testing for distributions was initiated by Chiesa and Gur~\cite{chiesa_et_al:LIPIcs.ITCS.2018.53}.
In addition to specific results discussed above, Herman and Rothblum also considered the computational complexity of the prover and in~\cite{Herman23_doubly_efficient} 
obtained {\em doubly efficient} proofs for label-invariant properties, in which both the prover and the verifier run in polynomial time, improving over their previous work in that the prover runs in polynomial time. Recent work \cite{biswas2025interactive} further shows that by augmenting the verifier with stronger query access (e.g., conditional sampling oracles), the sample complexity of prover-assisted distribution testing can be reduced exponentially to polylogarithmic in the domain size.

Several recent works have further developed the landscape of interactive proofs for distribution testing. \cite{GHR25-farness} study interactive proofs for $\epsilon$-\emph{farness}, the promise problem of distinguishing pairs of distributions that are $\epsilon$-far from pairs that are identical, as well as farness from uniformity. Their results give tradeoffs between the verifier's sample complexity and the honest prover's sample complexity, and are complementary to the label-invariant and quantity-verification problems that motivated this work. \cite{CPSTY25} revisits the Chiesa--Gur framework, showing that modifying the verification task, for example, by switching from verifying uniformity to verifying far-from-uniformity, can lead to striking improvements, including constant verifier sample complexity in some settings; they also initiate the study of multi-prover and zero-knowledge variants for distribution verification. These works explore new non-private proof-system formulations for distribution testing, whereas our focus is on how differential privacy constrains and enables prover-aided verification.

\section{Preliminaries}

The following is the definition of differential privacy, introduced in \cite{Dwork_McSherry_Nissim_Smith_2017}:
\begin{definition}
Datasets $D$ and $D'$ are \emph{adjacent databases} if $D$ and $D'$ differ by the addition or removal of one element.
\end{definition}

\begin{definition}
An interactive protocol $\mathcal{P}$ is $(\varepsilon,\delta)$-differentially private if for any pair of adjacent datasets $X,X'$, any adversary $A$, and any set of transcripts $E$, we have

$$\Pr[\Interact(\mathcal{P}, A, X) \in E] \le e^\varepsilon \cdot \Pr[\Interact(\mathcal{P}, A, X') \in E] + \delta$$

where $\Interact(\mathcal{P}, A, X)$ denotes the transcript output by the protocol $\mathcal{P}$ given some fixed (potentially adversarial) prover $A$, and an honest verifier that follows the protocol with sample $X$.
\end{definition}

One common DP mechanism is the Laplace mechanism. We say that a function $f$ has \emph{global sensitivity} $\Delta$ if on any adjacent datasets $X,X'$, $\left|f(X) - f(X')\right| \le \Delta$. It is known that the mechanism which releases $f(X) + \Lap\left(0,\frac{\Delta}{\varepsilon}\right)$ is $(\varepsilon,0)$-DP.

However, in cases where the global sensitivity is large while "most" datasets $X$ satisfy $\left|f(X) - f(X')\right| \le \Delta'$ for all adjacent datasets $X'$ and some $\Delta' \ll \Delta$, we would want to instead add noise scaled to $\Delta'$. We can do this with a technique called \emph{Propose-Test-Release} \cite{dwork2009differential} in which we propose a bound on the sensitivity, privately test whether this bound holds, and release the noisy estimate using noise scaled to the proposed bound rather than the global sensitivity if the test is passed.

In \emph{Interactive Proofs of Proximity}, $\mathcal{P}$ is an interaction between a prover and a verifier. The verifier wishes to determine whether a distribution $\mathcal{D}$, which they have sample access to, satisfies some property $P$. This protocol has proximity parameters $\sigma_1 < \sigma_2$ and satisfies:

\begin{itemize}
    \item \textbf{Completeness:} If the prover is honest and there exists some distribution $Q \in P$ such that $D_{TV}(\mathcal{D},Q) \le \sigma_1$, then the verifier accepts with high probability.
    \item \textbf{Soundness:} If there is no distribution $Q \in P$ such that $D_{TV}(\mathcal{D},Q) \le \sigma_2$, then the verifier rejects with high probability regardless of the prover's strategy.
\end{itemize}
When $\sigma_1>0$ we say the tester is {\em tolerant}.

In this work, we will be working with AM (Arthur-Merlin) protocols, also referred to as \emph{Public Coin Interactive Proofs}. In this setting, the only messages that the verifier is allowed to send to the prover are random bits, and the number of bits cannot depend on the dataset. In particular, these messages are completely independent of samples drawn by the verifier, so we can assume that all interactions occur prior to the samples being drawn \cite{chiesa_et_al:LIPIcs.ITCS.2018.53}.

We note that this is analogous to the public coin setting in property verification of functions, in which we have no secret information and the prover can see all the random coins we flip \cite{Goldwasser86}. In this setting, for any protocol, the prover getting every random coin we flip is sufficient for the prover to compute any message we could have sent them. When verifying properties of functions, there is no gap between the public coin and private coin setting \cite{Goldwasser86}. In contrast, there is provably a significant gap between public coin and private coin protocols for distribution property verification \cite{chiesa_et_al:LIPIcs.ITCS.2018.53}.

When reducing DP private coin protocols to DP AM protocols, we will use the following definitions to convert from DP to the form that we need for the reduction:

\begin{definition}
    (Max-Information \cite{dwork2015generalizationadaptivedataanalysis}). Let $X$ and $Z$ be jointly distributed random variables over the domain $(\mathcal{X},\mathcal{Z})$. The $\beta$-approximate max-information between $X$ and $Z$, denoted $I^\beta_\infty(X;Z)$, is defined as

    $$I^\beta_\infty(X;Z) = \log \sup_{\mathcal O\subseteq\mathcal X\times\mathcal Z; \Pr((X,Z) \in\mathcal O)>\beta} \frac{\Pr\left[(X,Z) \in\mathcal O\right]-\beta}{\Pr\left[X \otimes Z \in\mathcal O\right]}$$
    We say that an algorithm $\mathcal{A}:\mathcal{X}^s \rightarrow \mathcal{Y}$ has $\beta$-approximate max-information of $k$, denoted as $I^\beta_\infty(X;\mathcal{A}(X)) \le k$, if for every distribution $\mathcal{S}$ over elements of $\mathcal{X}^s$, we have $I^\beta_\infty(\mathbf{X};\mathcal{A}(X)) \le k$ when $\mathbf{X}\sim \mathcal{S}$. We denote by $I^\beta_{\infty,P}(\mathbf{X};\mathcal{A}(X)) \le k$ the same guarantee with the additional requirement that $\mathcal{S}$ be a product distribution.
\end{definition}

\begin{definition}
    (One-way perfect generalization \cite{bun2023stabilitystableconnectionsreplicability}). An algorithm $\mathcal{A}:\mathcal{X}^m \rightarrow \mathcal{Y}$ is said to be $(\beta,\varepsilon,\delta)$-one-way perfectly generalizing if for every distribution $D$ over $\mathcal{X}$, there exists a distribution $Sim_D$ such that with probability at least $1-\beta$ over the draw of an i.i.d. sample $S \sim D^m$, for every output set $O \subseteq \mathcal{Y}$ we have

    $$\Pr[\mathcal{A}(S) \in O] \le e^\varepsilon \Pr_{Sim_D}[O]+\delta.$$
\end{definition}

\begin{definition}
    (Replicability \cite{reproducibility}). Let $D$ be a distribution over domain $\mathcal{X}$. Let $\mathcal{A}$ be a randomized algorithm that takes as input samples from $D$. We say that $\mathcal{A}$ is $\rho$-reproducible if

    $$\Pr_{S,S',r}[\mathcal{A}(S,r) = \mathcal{A}(S',r)] \ge 1-\rho$$
    where $S,S'$ are sets of samples drawn i.i.d. from $D$ and $r$ represents the internal randomness of $\mathcal{A}$.
\end{definition}

Intuitively, if the distribution of outputs is similar when the algorithm is run on two samples drawn i.i.d. from the same distribution $D$, then we can make those output distributions agree on most random seeds $r$.

\paragraph{Miscellaneous Notation} At various points in this paper, we will write $f(x) \lesssim g(x)$ (resp. $\gtrsim, \asymp$) to indicate that two functions $f$ and $g$ satisfy $f(x) = O(g(x))$ (resp. $\Omega, \Theta$).
\section{Technical Overview}

\subsection{Lower Bound for Private Coin DP Interactive Proofs}

Our first result is a lower bound showing that given sufficiently small $(\varepsilon,\delta)$, anything we can do with $(\varepsilon,\delta)$-DP private coin interactive proofs can also be done with $(\varepsilon,\delta)$-DP public coin
protocols. Specifically,

\begin{theorem}
\label{thm: main reduction}
(Informal) Given a one-round $(\varepsilon,\delta)$-DP interactive proof verifying property $P$ with sample complexity $s$ and communication complexity $c$ with $\varepsilon = O\left(\frac{1}{\sqrt{s}}\right)$ and $\delta = O\left(\frac{1}{s^{5/2}}\right)$, there exists a one-round $(4\varepsilon,120\pi\delta)$-DP public coin interactive proof verifying property $P$ with sample complexity $s$ and communication complexity $c$.
\end{theorem}

These are the values of $(\varepsilon,\delta)$ which allow us to achieve $\rho$-replicability for constant $\rho$ through the transfer theorems of \cite{bun2023stabilitystableconnectionsreplicability}. When we have $\rho$-replicability for constant $\rho$, this tells us that given fixed randomness, the message sent from verifier to prover will be the same with high probability over the choice of sample.\footnote{Our proof likely still goes through with a bit more work if we don't go all the way to replicability, but doing so would only change constant factors.}

As such, we can simulate this interaction by letting the prover figure out (with high probability) an interaction that could have occurred with reasonable probability which would have resulted in the verifier accepting with high probability in the original protocol. A random seed $r$ is sent to the prover to "steer" them towards certain messages which may not be the message the verifier would have sent using random seed $r$ in the original private coin protocol. That is, the random seed $r$ steers the verifier towards accepting transcripts with certain initial messages with higher probability, and the prover (both honest and dishonest) will then also be steered towards those same messages, as they want the verifier to accept. The prover sends the verifier both what they think the verifier's message should have been as well as their response. The verifier then applies a soft threshold increasing with their probability of sending that message, their probability of accepting given this interaction in the original protocol, and how much the random seed r "steers" them towards this message.

Intuitively, the steering provided by the random seed $r$ ensures that if there is a message $m$ that would be sent with very high probability and results in acceptance, then it is very unlikely for it to steer the verifier away from accepting a transcript with $m$ as the first message. If no such message exists and the verifier would accept in the original protocol, there must be many messages which result in acceptance, and $r$ will most likely steer the verifier towards accepting at least one of them, which the prover will send as $m$. On the other hand, if the likelihood of acceptance in the original protocol is low, the messages which would result in acceptance must either be few, in which case it is unlikely that $r$ would steer the verifier to accept any of them, or it must be very unlikely for the verifier to send $m$ in the first place, in which case it is very unlikely for the steering to have enough of an effect to cause them to accept.

Furthermore, applying guarantees of privacy amplification by shuffling from  \cite{feldman2021hidingclonessimplenearly} allows us to achieve a similar guarantee for $\varepsilon$-local DP when $\varepsilon = O\left(\frac{1}{\sqrt{\log s}}\right)$, a much more relaxed privacy guarantee. Specifically, applying shuffling to  an $\varepsilon$-local DP mechanism for $\varepsilon = O\left(\frac{1}{\sqrt{\log s}}\right)$ gives us exactly the $(\varepsilon,\delta)$ above that achieves $\rho$-replicability for constant $\rho$.

\begin{remark}
The parameter realm $\varepsilon < \frac{1}{\sqrt{s}}$ is of great interest. In the context of counting queries, $\frac{1}{\sqrt{s}}$-DP can be achieved by adding noise of order $O(\sqrt{s})$, comparable to the sampling error.  By the advanced composition theorem~\cite{dwork2010boosting}, taking just a slightly smaller $\varepsilon = O\left(\frac{1}{\sqrt{s \ln (1/\delta)}}\right)$ defeats reconstruction attacks~\cite{dinurnissim} that use $\Theta(s)$ queries~\cite{dwork2007price,dwork2008new}, ensuring $(O(1), \delta)$-differential privacy overall.
\end{remark}

\noindent
Here are some examples of what Theorem~\ref{thm: main reduction} gives us:
\begin{enumerate}
    \item  To beat the $\tilde{O}(N^{2/3})$ sample complexity of \cite{herman2024public} for verifying label-invariant properties, we would have the following two constraints on $s$:
    $s \le 1/\varepsilon^2$ (this comes from $\eps \le 1/\sqrt{s}$ to apply our reduction) and
$s \le N^{2/3-\gamma}$ for some suitably small constant $\gamma$ (to beat the sample and communication complexity of \cite{herman2024public}).
Then if we have a private coin $(\varepsilon,\delta$)-DP protocol for property $P$, then there is a public coin DP protocol with the same privacy guarantees and communication and sample complexities.

    \item To beat our protocol for label-invariant properties in Section~\ref{sec:label_invariant}, it suffices to beat the trivial $O(N)$ communication protocol in which the prover sends a full description of the distribution $D$ to the verifier. This is because our sample and communication complexity bound of $\tilde{O}((N/\varepsilon)^{2/3})$ gives us a sample complexity of $\Omega(N)$ when $\varepsilon \le 1/\sqrt{s}$.

    \item Public coin lower bounds appear (to us) to be easier to prove than private coin lower bounds. See for example the \cite{chiesa_et_al:LIPIcs.ITCS.2018.53} lower bound which uses simulation to show that public coin protocols can do at most quadratically better than testers, where this simulation only works because of the lack of a secret held during the communication process. This result also gives us a way to convert lower bounds on DP AM protocols to lower bounds on DP private-coin protocols: if the communication complexity implied by an (hypothetical) lower bound is smaller than $N$ for $\varepsilon < 1/\sqrt{s}$, then this would imply the existence of a matching lower bound in the private coin setting.
\end{enumerate}

\subsection{Independence Testing}

The next result is a matching upper and lower bound for independence testing. The upper bound is achieved by a non-interactive proof system in which the prover sends the marginal distribution to the verifier, and then the verifier runs an identity tester to determine whether the distribution matches the unique product distribution defined by the marginals they were sent.

If the distribution is a product distribution and the prover is honest, then the marginals sent to the verifier are the true marginals of the distribution, and they run an identity tester against the true distribution, which will accept with high probability.

If the distribution is $\sigma$-far from any product distribution, it must be $\sigma$-far from the distribution defined by the marginals sent by the prover, regardless of what marginals the prover sends.

As such, independence testing can be solved with the same sample complexity as identity testing. Applying known bounds for testing identity with DP, we get the following theorem:

\begin{theorem} (Informal) 
    There exists an $\varepsilon$-DP non-interactive proof protocol $\mathcal{A}$ for verifying independence that has $O\left(\sqrt{N}/(\sigma\sqrt{\varepsilon})+\sqrt{N}/\sigma^2\right)$ sample complexity, $O\left(\sqrt{N}\log N/(\sigma\sqrt{\varepsilon})+\sqrt{N}\log N/\sigma^2\right)$ verifier runtime, and $O(\sum_{i=1}^d |A_i|)$ communication complexity with the following completeness and soundness guarantees:

    \textbf{Completeness:} If $D$ is a product distribution over $A_1 \times A_2 \times \dots \times A_d$ and the prover is honest, then the verifier accepts with probability at least $0.75$.\\

    \textbf{Soundness:} If $D$ is $\sigma$-far from any product distribution over $A_1 \times A_2 \times \dots \times A_d$, the verifier accepts with probability at most $0.25$ for any prover.
\end{theorem}

For the lower bound, we reduce uniformity testing to independence testing. Uniformity is known to have the same sample complexity as identity, so this gives us a matching lower bound up to log factors.

\begin{theorem}
    Suppose $A$ is a protocol for verifying independence with $s$ samples from a distribution $D$ over the space $A_1 \times A_2 \times \dots \times A_d$ and $C$ communication complexity. Further, suppose $A$ has the following completeness and soundness guarantees:

\textbf{Completeness}: If $D$ a product distribution and the prover is honest, then the verifier will accept with probability at least $0.9$.

\textbf{Soundness}: If $D$ is at least $\sigma$-far from a product distribution, then the verifier will accept with probability at most $0.1$ for any prover.

Then $s \ge \tilde{\Omega}\left(\frac{\sqrt{N}}{\sigma^2}\right)$
\end{theorem}

This is proved by mapping the uniform distribution onto $\{0,1\}^d$ and noting that a distribution is uniform iff it has uniform marginals and it is a product distribution. Verifying these two properties does not immediately solve uniformity as the closest product distribution to a distribution with uniform marginals may not be the uniform distribution. However, the closest product distribution in KL divergence to any distribution is the product distribution with the same marginals, and we can use Pinsker's inequality to relate total variation distance to KL divergence. Doing this achieves an $O(\sqrt{d})$ gap between the distance to the nearest product distribution and the distance to the product distribution with the same marginals, provided the marginals are close to uniform. Since $d = \log N$, this ends up only contributing log factors and we get a matching lower bound.

\subsection{Generic Reduction for \texorpdfstring{$\varepsilon$}{Epsilon}-DP AM Protocols}

The next result is a generic reduction that converts any non-private AM protocol to an $\varepsilon$-DP AM protocol with an $O(1/\varepsilon)$ blowup. 

\begin{theorem}
Suppose AM distribution tester $T$ can test whether some distribution $D$ $\xi$-approximately satisfies $\mathcal{P}$ for some property $\mathcal{P}$ with probability $\frac{5}{6}$. Suppose further that $T$ has sample complexity $s$ and communication complexity $c$. There exists an $\varepsilon$-DP AM tester $T'$ that can test whether $D$ $\xi$-approximately satisfies $\mathcal{P}$ with probability $\frac{2}{3}$ using sample complexity $\left\lceil\frac{6}{\varepsilon}\right\rceil s$ and communication complexity $c$.
\end{theorem}

The proof of this uses the same proof technique as a similar result in \cite{aliakbarpour2017differentiallyprivateidentitycloseness, aliakbarpour2017differentiallyprivateidentityclosenessArxiv}\footnote{As mentioned before, the reduction is in the full version \cite{aliakbarpour2017differentiallyprivateidentityclosenessArxiv}.} which converts non-private testers (without a prover) to $\varepsilon$-DP testers with an $O(1/\varepsilon)$ blowup. The key idea to this result is noting that since messages in AM protocols cannot depend on the sample, we can first complete all of the communication and then draw the sample and treat the decision part of the protocol as a tester, thus allowing us to conclude our proof in the same way as they did.

\subsection{Improved Bounds for Label-Invariant Properties with Propose-Test-Release}

The next result improves upon this generic result by showing that we do not need to pay an $O(1/\varepsilon)$ factor in the sample complexity when verifying label-invariant properties. Instead, we adopt the propose-test-release framework to modify the AM protocol of \cite{herman2024public}, described in Algorithm~\ref{alg:non-private-label-invariant} for privacy, achieving a multiplicative blowup of $O(\varepsilon^{2/3})$ rather than $O(\varepsilon)$.

Specifically, we introduce a protocol, Algorithm~\ref{alg:label-invariant}, which modifies Algorithm~\ref{alg:non-private-label-invariant} by adding Laplace noise scaled to the sensitivity of each check involving the sample.  However, one of these checks counts 3-way collisions among samples, which has sensitivity depending on the dataset, so we must use propose-test-release to bound the sensitivity of this check. Additionally, the introduction of DP means the weight of heavy elements must also be measured with DP, and as such, the protocol must be able to handle elements with weight up to $\log(1/\delta)/\varepsilon s$ for $(\varepsilon,\delta)$-DP and sample size $s$ as opposed to $1/s$ in the non-private protocol. This change results in a non-trivial modification in the technical analysis of the DP protocol. 

\begin{theorem}
    Given privacy parameters $\eps,\delta > 0$ and a proximity parameter $\sigma > 0$ that are not too small (i.e. satisfying $\log(1/\delta)/\eps\sigma \le N^{o(1)}$), along with access to $s = \tO(N^{2/3}/\eps^{2/3}) \cdot \mathrm{poly}(\log(1/\delta)/\sigma)$ samples from $D$, \Cref{alg:label-invariant} satisfies the following completeness and soundness conditions:
    \begin{itemize}
        \item \textbf{Completeness:} If the prover is honest, then with probability at least $3/4$, the verifier accepts and outputs tagged samples $(S_1, \pi_1), \ldots, (S_s, \pi_s)$ satisfying
        \[
            \frac{1}{s}\sum_{i \in [s] : \pi_i \ge \frac{\sigma}{N}} \biggl(1 - \min\Bigl\{\frac{\pi_i}{D(S_i)}, \frac{D(S_i)}{\pi_i}\Bigr\}\biggr) \le O(\sigma^3).
        \]
        and
        \[
            \frac{1}{s}\sum_{i\in[s]:\pi_i<\frac{\sigma}{N}} D(S_i) \le O\Bigl(\frac{\sigma}{N}\Bigr).
        \]
        \item \textbf{Soundness:} If the prover is dishonest, then with probability at least $3/4$, the verifier either rejects or outputs tagged samples $(S_1, \pi_1), \ldots, (S_s, \pi_s)$ satisfying
        \[
            \frac{1}{s}\sum_{i \in [s] : \pi_i \ge \frac{\sigma}{N}} \biggl(1 - \min\Bigl\{\frac{\pi_i}{D(S_i)}, \frac{D(S_i)}{\pi_i}\Bigr\}\biggr) \le O(\sigma).
        \]
        and
        \[
            \frac{1}{s}\sum_{i\in[s]:\pi_i<\frac{\sigma}{N}} D(S_i) \le O\Bigl(\frac{\sigma}{N}\Bigr).
        \]
    \end{itemize}
    (Note that the right-hand side of the first inequality in the completeness guarantee is $O(\sigma^3)$, but the corresponding quantity in the soundness guarantee is $O(\sigma)$. The sets of inequalities are otherwise identical up to suppressed constant factors.)
\end{theorem}

This states that the protocol gives us \emph{tagged samples}, or ordered pairs $(S_i,\pi_i)$ of elements $S_i \in [N]$ and their weights $\pi_i \approx D(S_i)$ for distribution $D$. \cite{herman2024public} shows how to verify any label-invariant property with an approximate histogram computed using these tagged samples entirely in postprocessing (i.e. not using the verifier's sample), and we obtain tagged samples with the same guarantees, so black boxing their result allows us to achieve private verification of any label-invariant property.

\section{Lower Bound for Private Coin DP Interactive Proofs}
\label{sec:private-to-public}

In this section, we provide a reduction from one-round differentially private private-coin Interactive Proofs for sufficiently small $\varepsilon$ to differentially private public-coin (AM) protocols by converting differential privacy to replicability through the transfer theorems in \cite{bun2023stabilitystableconnectionsreplicability}.

More concretely, Corollary 3.12 of \cite{rogers2016maxinformationdifferentialprivacypostselection} converts DP to bounded max-info at a cost of a multiplicative $\sqrt{s}$ blow-up in the max-info for sample size $s$, and then Lemma 3.14 and Theorem 3.17 of \cite{bun2023stabilitystableconnectionsreplicability} converts this guarantee to one-way perfect generalization and then replicability at no additional cost. 

Then, if we are working with replicable protocols, we can guarantee that the message sent from the verifier to the prover will be the same with high probability over the choice of sample. As such, we can simply remove the sample from the equation altogether and just send the random seed, which is sufficient to determine the message with high probability.

We then let the prover determine both the message we would have sent as well as their response. We accept this transcript with probability scaling up with the probability that we would have sent the same message and the probability that we would accept had this been the transcript in the original protocol. This soft threshold allows us to combine replicability with the original completeness and soundness guarantees to get a completeness and soundness of our new protocol while preserving privacy.

This reduction is not computationally efficient, as the goal is solely to ensure that the sample and communication complexity remain unchanged. 

\begin{theorem}
    Suppose $\mathcal{A}(S,r)$ is a $\rho$-replicable, $(\varepsilon,\delta)$-DP algorithm, $\mathcal{D}$ is the population distribution, and $\mathcal{R}$ is a distribution over random seeds. Additionally, suppose we have the following (private coin) protocol with sample complexity $s$ and communication complexity $C$:

    \begin{enumerate}
        \item The Verifier draws sample $S \sim \mathcal{D}^s$ and randomness $r \sim \mathcal{R}$ and sends $m = \mathcal{A}(S,r)$ to the prover.
        \item the Prover responds with $m' = f(m,\mathcal{D})$.
        \item The Verifier computes $\mathcal{A}'(T,S,r, r')$ where $T = (m,m')=(\mathcal{A}(S,r),m')$ is the transcript of the interaction and $r'$ is any new randomness drawn for $\mathcal{A}'$, which outputs $1$ to accept and $0$ to reject.
    \end{enumerate}

    satisfying the following completeness and soundness guarantees:

    \textbf{Completeness:} if $T$ was generated according to the protocol (that is, if $T = (\mathcal{A}(S,r),f(m,\mathcal{D}))=(\mathcal{A}(S,r),f(\mathcal{A}(S,r),\mathcal{D}))$) and $\mathcal{D}$ satisfies the property we wish to verify, then

$$\Pr_{S \sim \mathcal{D}^s,r \sim \mathcal{R}, r' \sim \mathcal{R}'}[\mathcal{A}'(T,S,r,r') = 1] \ge 0.999.$$

\textbf{Soundness:} if $\mathcal{D}$ is $\xi$-far from the property we wish to verify, then for any function $f'$ generating $m' = f'(m,\mathcal{D})=f'(\mathcal{A}(S,r),\mathcal{D})$ and $T = (\mathcal{A}(S,r),m')$,

$$\Pr_{S \sim \mathcal{D}^s,r \sim \mathcal{R},r' \sim \mathcal{R}'}[\mathcal{A}'(T,S,r,r') = 1] \le 0.001.$$

Then, there exists a $(4\varepsilon,120\pi\delta)$-DP\footnote{Standard privacy amplification by subsampling arguments allow this constant factor blowup in the privacy parameters to instead be replaced with constant factor blowups in the sample complexity.} AM protocol in which the verifier sends only random string $r \sim \mathcal{R}$ to the prover satisfying the following completeness and soundness guarantees for $\rho \le 10^{-4}$:

\textbf{Completeness:} if $T$ was generated according to the protocol (that is, if the prover was honest) and $\mathcal{D}$ satisfies the property we wish to verify, then the verifier accepts with probability at least $0.733$.

\textbf{Soundness:} if $\mathcal{D}$ is $\xi$-far from the property we wish to verify, for any function $f'$ generating $T = f'(r,\mathcal{D})$, then the verifier accepts with probability at most $0.0943$.

Furthermore, the sample complexity of this protocol is $s$, and the communication complexity is $O(C) + \log|\mathcal{R}|$ where $C$ is the communication complexity of the original protocol and $|\mathcal{R}|$ is the support size of the random seeds.

\end{theorem}

\begin{proof}
    For a given dataset \(S\), we introduce the following notation:
\[
p_S(m)=\Pr_{r \sim \mathcal{R}}[\text{original verifier sends }m],
\]
\[
q_S(m,a)=
\Pr_{r \sim \mathcal{R},r' \sim \mathcal{R}'}[\text{original verifier sends }m
\text{ and accepts if the reply is }a],
\]
\[
b_S(m,a)=q_S(m,a)/p_S(m) = \Pr_{r \sim \mathcal{R},r' \sim \mathcal{R}'}[\text{original verifier accepts} \mid \text{original verifier sends }m],
\]
with value \(0\) when \(p_S(m)=0\).

Additionally, let
\[
\mu(m)=\mathbb E_{S\sim D^s}p_S(m).
\]

Here is the new protocol.

1. The verifier sends a seed \(r \sim \mathcal{R}\) and defines pairwise-independent random variables
   \[
   U_m(r)\sim\operatorname{Unif}\{1/L,\ldots,1\},
   \qquad m\in\mathcal M,
   \]
   where \(L\ge|\mathcal M|\) is the smallest such power of two. An ordinary pairwise-independent hash family uses \(O(\log|\mathcal M|) = O(C)\) public bits so this takes $O(\max \{\log |\mathcal{R}|, C\})$ bits of communication as if the original random seeds were long enough, we can just use them to define a pairwise-independent hash family. If not, then we just pad them out to length $O(C)$ for the hash family.

2. The prover returns a claimed original transcript \((m,a)\).

3. The verifier accepts with probability
   \[
   \boxed{\quad
   \frac{\kappa^2p_S(m)\,q_S(m,a)}
        {U_m(r)^2+\kappa^2p_S(m)^2}
   \quad} \tag{2}
   \]
   for a fixed constant \(\kappa\).

Equivalently, this is the original conditional acceptance probability \(b_S(m,a)\), multiplied by the smooth consistency gate
\[
G(p,u)=\frac{\kappa^2p^2}{u^2+\kappa^2p^2}.
\]
where we will be using $p = p_S(m)$ and $u = U_m(r)$. 

To establish completeness, fix another constant \(t>1\). The honest prover chooses any message satisfying
\[
U_m(r)\le t\mu(m),
\]
and supplies the original honest reply \(a=f(m,D)\).

Such a message exists with probability at least \(1-1/t\). Indeed, if some \(t\mu(m)\ge1\), existence is automatic. Otherwise the number of qualifying messages has expectation at least \(t-1\); pairwise independence and the second-moment bound give failure probability at most \(1/t\). To see why this is the case:

$$\Pr_{r \sim \mathcal{R}}\left[U_m(r) \le t\mu(m)\right]= \frac{\lfloor L t\mu(m)\rfloor}{L} \ge t\mu(m) - \frac{1}{L}$$

and $L \ge |\mathcal{M}|$ so if we let $Z$ be the number of messages $m$ such that $U_m(r) \le t\mu(m)$,

$$\E_{r \sim \mathcal{R}}[Z] =\sum_{m \in \mathcal{M}} \Pr_{r \sim \mathcal{R}}\left[U_m(r) \le t\mu(m)\right] \ge t - 1$$

Additionally, for each $m$, the random variable for whether $U_m(r) \le t\mu(m)$ is a Bernoulli random variable. In particular, this means its variance is less than its mean. Since the pairwise independent hash guarantees independence of these random variables, the variance of $Z$ is just the sum of the variances of these Bernoullis, and is thus at most $\E_{r \sim \mathcal{R}}[Z]$. Applying the one-sided Chebyshev inequality then gives us

$$\Pr[Z \le 0] \le \frac{\Var(Z)}{\Var(Z) + \E[Z]^2} \le \frac{1}{1+\E[Z]} \le \frac{1}{t}.$$



Suppose the original protocol has completeness error and soundness error at most \(\alpha\). When the prover is honest and $D$ satisfies the property, we therefore have $\E_{S\sim D^s, r \sim \mathcal{R}} \sum_{m \in \mathcal{M}}q_S(m,a) \ge 1-\alpha$. Additionally, replicability bounds the distance between the distribution over messages defined by $\mu$ and $p_S$. This is bounded by $\rho$ as if we fix some sample $S$ and then compute $\mathcal{A}(S,r)$ for $r \sim \mathcal{R}$, the distribution over messages would be $p_S$. On the other hand, if we took a fresh sample $S'$ and computed $\mathcal{A}(S',r)$ for $r \sim \mathcal{R}$, the distribution over messages would be $\mu$. Therefore, replicability gives us

$$\E_S d_{TV}(p_S, \mu) \le \Pr_{S,S' \sim \mathcal{D}^s,r \sim \mathcal{R}}\left[\mathcal{A}(S,r) \ne \mathcal{A}(S',r)\right] \le \rho$$

We now consider our failure conditions:
\begin{enumerate}
    \item We receive a good message and accept with probability at least $\frac{\kappa^2}{\kappa^2+t^2}$
    \item The honest prover fails to find a message to send: this happens with probability $1/t$.
    \item The honest prover sends a message $m$ such that $U_m(r) \le t\mu(m)$ but $U_m(r) \ge tp_S(m)$, causing the probability of rejection to exceed $\frac{\kappa^2}{\kappa^2+t^2}$.
\end{enumerate}

Point 1 is independent of the others but the others may be correlated so we compute the additional error probability that arises from the distance between the relevant quantities when $m$ is drawn from $\mu$ vs $p_S$ and take the appropriate union bounds:

Note that
$$\E_{r \sim \mathcal{R}} G(p, U_m(r)) \le \int_0^1 G(p,u) du = \kappa p \cot^{-1}(\kappa p) \le \kappa  p  \pi.$$
\footnote{Inverse tangent and cotangent can be defined with range being either $(-\pi/2,\pi/2]$ or $(0,\pi]$. For the sake of correctness with respect to both, we upper bound its value at $\pi$.}

Further, suppose WLOG that $p_S(m) \ge \mu(m)$ (the argument for the other direction follows symmetrically since we took an absolute value). Then,
$$\E_{r \sim \mathcal{R}}\left|G(p_S(m),U_m(r)) - G(\mu(m), U_m(r))\right| = \int_{\mu(m)}^{p_S(m)}\E_u\left[\frac{\partial G}{\partial x}(x,u) dx\right]$$

and
$$\frac{\partial G}{\partial x}(x,u) = \frac{2\kappa^2 x u^2}{(u^2+\kappa^2x^2)^2} \le \frac{2\kappa^2 x}{u^2+\kappa^2 x^2}$$

which again only has $u$ show up in the denominator, so an upper bound on $u$ again lower bounds the value of interest. We then note that

$$\int_0^1 \frac{2\kappa^2x}{u^2+\kappa^2x^2}du = 2\kappa \cot^{-1}(\kappa x) \le 2\pi \kappa.$$

Thus,

\begin{align*}\E_{r \sim \mathcal{R}}\left|G(p_S(m),U_m(r)) - G(\mu(m), U_m(r))\right| &\le\int_{\mu(m)}^{p_S(m)} 2\pi\kappa dx = 2\pi\kappa(p_S(m)-\mu(m)).
\end{align*}

so the total variation bound gives us
$$\sum_m\E_{S,r}\left|G(\mu(m), U_m(r)) - G(p_S(m), U_m(r))\right| \le 4\kappa \rho \pi$$

Now, recall that $\Pr_r[U_m(r) \le t\mu(m)] \le t\mu(m) = t\Pr_{S,r}[\text{Send message $m$ in original protocol}]$. This means that the probability of sending a message $m$ that would be rejected is multiplied by at most $t$, and sending messages from $\mu$ instead of $p_S$ increases the rate of rejection to at most $\alpha+\rho$. Putting these together, we have

$$\E_{S,r}[\mathbf{1}_{\text{Honest Prover sends a message}} \cdot (1-b_S(m,f(m)))] \le t(\alpha+\rho).$$

Thus,

\begin{align*}
\Pr_{S \sim \mathcal{D}^s, r \sim \mathcal{R}}[\mathrm{accept}] &\ge \left(1-\frac{1}{t} - t(\alpha+\rho)\right)\left(\frac{\kappa^2}{\kappa^2+t^2}\right)- (4\kappa \rho \pi)
\end{align*}

For soundness, we instead have $\sum_{m \in \mathcal{M}} \max_a \E_{S}[q_S(m,a)] \le \alpha$. Since the sample $S$ can be drawn after interaction with the prover, we can bound the probability that we accept as follows:
\[
\E_{r \sim \mathcal{R}}\left[\max_{m,a} \E_{S}\left[G(p_S(m), U_m(r)) \cdot b_S(m,a)\right]\right]
\]

Now, note that $b_S(m,a) \le 1$. Therefore,

\begin{align*}
    &\max_{m,a} \E_{S}\left[G(p_S(m), U_m(r)) \cdot b_S(m,a)\right]\\
    &= \max_{m,a} \left[\E_{S}\left[G(\mu(m), U_m(r)) \cdot b_S(m,a)\right] + \E_{S}\left[\left(G(p_S(m),U_m(r)) - G(\mu(m), U_m(r))\right) \cdot b_S(m,a)\right]\right]\\
    &\le \sum_{m \in \mathcal{M}} \left[G(\mu(m), U_m(r)) \max_a \E_{S}[b_S(m,a)] + \E_{S}\left[\left|G(p_S(m),U_m(r)) - G(\mu(m), U_m(r))\right|\right]\right]
\end{align*}

where the final inequality comes from noting that $b_S(m,a)$ is a probability and is thus at most $1$, and we can upper bound a max with a sum.

Recall that 
$$\E_{r \sim \mathcal{R}} G(p, U_m(r)) \le \kappa  p  \pi.$$

We thus have

\begin{align*}\E_{r\sim \mathcal{R}}\left[\sum_{m \in \mathcal{M}} G(\mu(m), U_m(r)) \max_a \E_{S}[b_S(m,a)]\right] &\le \sum_{m \in \mathcal{M}}\kappa \pi \mu(m)\max_a\E_{S}[b_S(m,a)]\\
&=\kappa \pi\sum_{m \in \mathcal{M}} \max_a \E_{S}[p_S(m)  b_S(m,a) + (\mu(m) - p_S(m)) b_S(m,a)]\\
&\le \kappa \pi (\alpha + \rho)\end{align*}

where the final inequality comes from soundness of the original protocol and the fact that $b_S(m,a) \le 1$ so we can bound $\sum_{m \in \mathcal{M}} \max_a(\mu(m) - p_S(m)) b_S(m,a)$ by $\E_S d_{TV}(\mu, p_S) \le \rho$ as shown above.

For the other term, recall that

\begin{align*}\sum_{m}\E_{S \sim D^s,r \sim \mathcal{R}}\left|G(p_S(m),U_m(r)) - G(\mu(m), U_m(r))\right| &\le 4\pi\kappa\rho.
\end{align*}

so our probability of acceptance is at most

$$\pi\kappa (\alpha+\rho) + 4\kappa\pi\rho$$


For concrete constants,
\[
\alpha=0.001,\qquad \rho\le10^{-4},\qquad t=6,\qquad\kappa=20
\]
give
\[
\Pr[\mathrm{accept}\mid\mathrm{yes}]\ge0.733,
\qquad
\Pr[\mathrm{accept}\mid\mathrm{no}]\le0.0943.
\]

We now prove privacy: for pure DP, this is fairly direct. For neighboring datasets, the original joint probabilities
\[
q_S(m,a),\qquad p_S(m)-q_S(m,a), \text{ and } \qquad p_S(m)
\]
and change multiplicatively by at most \(e^\varepsilon\). Consequently:

\begin{itemize}
\item the conditional accept/reject distribution 
$$b_S(m,a) = q_S(m,a)/p_S(m), \qquad 1-b_S(m,a) = (p_S(m)-q_S(m,a))/p_S(m)$$ is \(2\varepsilon\)-DP;
\item $G(p_S(m),U_m(r))$ \(2\varepsilon\)-DP because $U_m(r)$ is independent of $S$, and if we let $U_m(r) = y\kappa p_S(m),$ we have
$$\frac{e^{2\epsilon}}{y^2+e^{2\epsilon}} \le e^{2\epsilon} \cdot \frac{1}{y^2+1}$$
and
$$\frac{e^{-2\epsilon}}{y^2+e^{-2\epsilon}} \ge e^{-2\epsilon} \cdot \frac{1}{y^2+1}$$
and similarly for $1-G(p_S(m),U_m(r)) = \frac{y^2}{y^2+1}$
\end{itemize}
Since the accept/reject decision is simply the AND of these two bits, this gets us \(4\varepsilon\)-DP.

For approximate DP, we define "distributions" $\tilde{p}_S, \tilde{q}_S$ which are close to $p_S$ and $q_S$ in TV distance and multiplicatively close to $p_{S'}$ and $q_{S'}$ by shaving off the parts of the distribution that fall outside the multiplicative $e^\varepsilon$ window.

Specifically, fix neighboring datasets $S,S'$. Define $\tilde{p}_S(m) = \min\{e^{\varepsilon} p_{S'}(m), \max(e^{-\varepsilon} p_{S'}(m), p_S(m))$. For $q_S(m,a)$, note that we need to bound both $\frac{q_S(m,a)}{q_{S'}(m,a)}$ as well as $\frac{p_S(m) - q_S(m,a)}{p_{S'}(m) - q_{S'}(m,a)}$ so we instead clip it to the interval

$$I_{m,a} = \left[\max\{e^{-\varepsilon}q_{S'}(m,a), \tilde{p}_{S}(m) - e^{\varepsilon}(p_{S'}(m)-q_{S'}(m,a))\}, \min\{e^{\varepsilon}q_{S'}(m,a), \tilde{p}_{S}(m) - e^{-\varepsilon}(p_{S'}(m)-q_{S'}(m,a))\}\right]$$

This interval is non-empty because
$$\tilde{p}_S(m) - e^{\varepsilon}(p_{S'}(m)-q_{S'}(m,a)) = e^{\varepsilon}q_{S'}(m,a) - (e^{\varepsilon}p_{S'}(m)-\tilde{p}_S(m)) \le e^{\varepsilon}q_{S'}(m,a)$$
and
$$\tilde{p}_{S}(m) - e^{-\varepsilon}(p_{S'}(m)-q_{S'}(m,a)) = e^{-\varepsilon}q_{S'}(m,a) + (\tilde{p}_S(m) - e^{-\varepsilon}p_{S'}(m)) \ge e^{-\varepsilon}q_{S'}(m,a)$$
so this is well defined. Now, we note that
$$\sum_{m \in \mathcal{M}} \left|p_S(m) - \tilde{p}_S(m)\right| \le 2\delta$$
by definition of $(\varepsilon,\delta)$-DP, and for $\tilde{q}$ note that we're equivalently clipping to the intersection of the intervals $[e^{-\varepsilon} q_{S'}(m,a), e^\varepsilon q_{S'}(m,a)]$ and $[\tilde{p}_S(m) - e^{\varepsilon}(p_{S'}(m) - q_{S'}(m,a)), \tilde{p}_S(m) - e^{-\varepsilon}(p_{S'}(m) - q_{S'}(m,a))]$. By our DP guarantee, clipping to each of these intervals individually changes mass by at most $2\delta$ over all messages $m$ and any responses $a$, so clipping to the intersection changes the mass by at most $4\delta$ mass by a union bound. That is,

$$\sum_{m \in \mathcal{M}} \max_a \left| q_S(m,a) - \tilde{q}_S(m,a)\right| \le 4\delta.$$

By construction, the multiplicative gap between $\tilde{p}_S$ and $p_{S'}$ is at most $\varepsilon$ and similarly for $\tilde{q}_S$ versus $q_{S'}$ and $\tilde{p}_S - \tilde{q}_S$ versus $p_{S'} - q_{S'}$. Thus, to show that we have $(4\varepsilon,6\pi\kappa\delta)$-DP, it suffices to show that the acceptance probability changes by at most $6\pi\kappa\delta$ when we replace $p_S$ and $q_S$ with $\tilde{p}_S$ and $\tilde{q}_S$:

$$\E_r\max_{m,a}\left|\frac{\kappa^2 \tilde{p}_S(m) \tilde{q}_S(m,a)}{U_m(r)^2+\kappa^2 \tilde{p}_S(m)^2} - \frac{\kappa^2 p_S(m) q_S(m,a)}{U_m(r)^2+\kappa^2 p_S(m)^2}\right|$$

To do this, let
$$F(p,q,u) = \frac{\kappa^2 pq}{u^2+\kappa^2 p^2}$$
or equivalently,
$$F(p,q,u) = \frac{pq}{u^2/\kappa^2+p^2}$$

We can now compute:

$$\left|\frac{\partial F(p,q,u)}{dq}\right|= \left|\frac{p}{u^2/\kappa^2+p^2}\right| = \frac{p}{u^2/\kappa^2+p^2}$$
and
$$\left|\frac{\partial F(p,q,u)}{dp}\right| = \left|\frac{\kappa^2 q (u^2 - \kappa^2 p^2)}{(\kappa^2 p^2 + u^2)^2} \right|\le \frac{q}{u^2/\kappa^2 + p^2} \le \frac{p}{u^2/\kappa^2+p^2}$$

where the last inequality comes from the fact that by definition $q_S(m,a) \le p_S(m)$.

We can thus use the chain rule to bound the rate of change of a smooth interpolation between $p_S(m)$
and 
$q_S(m,a)$ to $\tilde{p}_S(m)$ and $\tilde{q}_S(m,a)$ as $F(p_t, q_t, u)$ where $p_t = (1-t)p_S(m)+t\tilde{p}_S(m)$ and 
$q_t = (1-t)q_S(m,a)+t\tilde{q}_S(m,a)$:

$$\frac{d}{d t} F(p_t, q_t,  u) = \frac{dF(p_t,q_t,u)}{dp_t} \frac{dp_t}{dt} + \frac{dF(p_t,q_t,u)}{dq_t} \frac{dq_t}{dt}$$

The desired difference is then:
\begin{align*}
    \E_r\max_{m,a} \left|\int_0^1 \frac{d}{dt}F(p_t,q_t,U_m(r)) dt\right| &\le \E_r\Biggl[\sum_{m \in \mathcal{M}}|p_S(m) - \tilde{p}_S(m)| \cdot\int_0^1 \frac{p_t}{U_m(r)^2/\kappa^2+p_t^2}dt\\
    &+ \sum_{m \in \mathcal{M}}\max_a |q_S(m,a) - \tilde{q}_S(m,a)| \cdot \int_0^1 \frac{p_t}{U_m(r)^2/\kappa^2+p_t^2}dt\Biggr]\\
\end{align*}

Fubini's Theorem allows us to change the order of integration/summation over t, r, and m, giving us

\begin{align*}&\E_r\left[\sum_{m \in \mathcal{M}}|p_S(m) - \tilde{p}_S(m)|\left( \int_0^1 \frac{p_t}{U_m(r)^2/\kappa^2+p_t^2}dt\right) + \sum_{m \in \mathcal{M}} \max_a|q_S(m,a) - \tilde{q}_S(m,a)| \left(\int_0^1 \frac{p_t}{U_m(r)^2/\kappa^2+p_t^2}dt\right)\right]\\
&= \sum_{m \in \mathcal{M}}|p_S(m) - \tilde{p}_S(m)|\left( \int_0^1 \E_r\left[\frac{p_t}{U_m(r)^2/\kappa^2+p_t^2}\right]dt\right)\\
&+ \sum_{m \in \mathcal{M}} \max_a|q_S(m,a) - \tilde{q}_S(m,a)| \left(\int_0^1 \E_r\left[\frac{p_t}{U_m(r)^2/\kappa^2+p_t^2}\right]dt\right)\end{align*}

Note that since we're summing over $m$ first now, the $m$ of $U_m(r)$ is fixed inside the expectation with respect to $r$ meaning $U_m(r)$ is uniform when taking the randomness over $r$ so we can once again use the bound:

$$\E_r\left[\frac{p_t}{U_m(r)^2/\kappa^2+p_t^2}\right]\le \int_0^1 \frac{p_t}{u^2/\kappa^2+p^2_t}du = \kappa \cot^{-1}(\kappa p_t) \le \kappa \pi$$

Combining this with the fact that 

$$\sum_{m \in \mathcal{M}}(p_S(m) - \tilde{p}_S(m)) \le 2\delta$$
and
$$\max_a\sum_{m \in \mathcal{M}}(q_S(m,a) - \tilde{q}_S(m,a)) \le 4\delta$$

we get that the probability of acceptance changes by at most $6\pi\kappa\delta$ when we use $\tilde{p}_S$ and $\tilde{q}_S$ instead of $p_S$ and $q_S$

Putting this all together and letting $\kappa = 20$ as above, this protocol is

$$\boxed{(4\varepsilon,120\pi\delta)\text{-DP}}$$
\end{proof}

We can combine this with the following lemmas to get our main result:

\begin{lemma}
\label{lem:DP_max_info}
    (Corollary 3.12 of \cite{rogers2016maxinformationdifferentialprivacypostselection}). Fix $m \in \mathbb{N}$, sufficiently small $\rho \in (0,1)$. Let $\varepsilon = \frac{\rho}{\sqrt{8m\log(1/\rho)}}$, $\delta \le \frac{\varepsilon \rho^6}{m^2}$. Let $\mathcal{A}:\mathcal{X}^m \rightarrow \mathcal{Y}$ be an $(\varepsilon,\delta)$-DP algorithm. Then, for any distribution $D$ over $\mathcal{X}$, if $S \sim D^m$, $I_\infty^{\rho^3}(S;\mathcal{A}(S)) \le O(\rho)$.
\end{lemma}

\begin{lemma}
\label{lem:max_info_generalizability}
    (Lemma 3.14 of \cite{bun2023stabilitystableconnectionsreplicability}). Fix $m \in \mathbb{N}$, $k > 0$, $\beta \in (0,1)$ and $\hat{\beta} = \sqrt{\frac{\beta}{1-2^{-k}}}$. Let $\mathcal{A}:\mathcal{X}^m \rightarrow \mathcal{Y}$ be an algorithm. Then, for every distribution $D$ over $\mathcal{X}$ and $S \sim D^m$, if $I^\beta_\infty(S;\mathcal{A}(S)) \le k$, then $\mathcal{A}$ is $(\hat{\beta},2k,\hat{\beta})$-one way perfectly generalizing.
\end{lemma}

\begin{lemma}
\label{lem:generalizability_replicability}
    (Theorem 3.17 of \cite{bun2023stabilitystableconnectionsreplicability}, restated).
    Fix $m \in \mathbb{N}$ and $\beta,\varepsilon,\delta \in (0,1)$. Let $\mathcal{A}:\mathcal{X}^m \rightarrow \mathcal{Y}$ be a $(\beta,\varepsilon,\delta)$-one way perfectly generalizing algorithm with finite output space. Then, for any distribution $D$ over $\mathcal{X}$, there exists an algorithm $\mathcal{A}'$ with identical output distribution to $\mathcal{A}$ that is $4(\beta+2\varepsilon+\delta)$-replicable.
\end{lemma}

\begin{corollary}
    Given a one-round $(\varepsilon,\delta)$-DP interactive proof verifying property $P$ with sample complexity $s$ and communication complexity $c$ with $\varepsilon = O\left(\frac{1}{\sqrt{s}}\right)$ and $\delta = O\left(\frac{1}{s^{5/2}}\right)$, there exists a one-round $(4\varepsilon,120\pi\delta)$-DP public coin interactive proof verifying property $P$ with sample complexity $s$ and communication complexity $c$.
\end{corollary}

\begin{proof}
    If we plug these values of $\varepsilon$ and $\delta$ into Lemmas~\ref{lem:DP_max_info}, \ref{lem:max_info_generalizability}, and \ref{lem:generalizability_replicability}, we can get a replicable algorithm with $\rho < 0.1$
\end{proof}

Using the following bound for privacy amplification by shuffling, we can additionally get a much tighter bound if we require the messages to be sent with local DP:

\begin{definition}
    An algorithm $\mathcal{A}$ satisfies $\varepsilon$-local DP if for any dataset $X = (X_1, X_2, \dots, X_s)$, $\mathcal{A}(X) = (\mathcal{A}_1(X_1), \mathcal{A}_2(X_2), \dots, \mathcal{A}_s(X_s))$ where $\mathcal{A}_i(X_i)$ satisfies $\varepsilon$-DP for every $i$.
\end{definition}

\begin{lemma}
    (Theorem 3.1 from \cite{feldman2021hidingclonessimplenearly}, simplified) Suppose $\mathcal{R}$ satisfies $\varepsilon_0$ local DP for $\varepsilon_0 \le 1$. Then, the algorithm $\mathcal{A}$ which generates a random permutation $\pi$ of $[s]$ and then outputs
    $$\mathcal{A}(X) = \left(\mathcal{R}_{1}(X_{\pi(1)}), \mathcal{R}_{2}(X_{\pi(2)}), \dots, \mathcal{R}_{s}(X_{\pi(s)})\right)$$
    satisfies $(\varepsilon,\delta)$-DP where

    $$\varepsilon = O\left(\varepsilon_0\sqrt{\frac{\log(1/\delta)}{s}}\right)$$
\end{lemma}

\begin{corollary}
    If there exists a one-round private-coin interactive proof for verifying distribution property $\mathcal{P}$ with sample complexity $s$ and communication complexity $C$ such that the algorithm the verifier uses to compute its message satisfies $\varepsilon$-local DP for $\varepsilon = O\left(\frac{1}{\sqrt{\log s}}\right)$, then there exists a public-coin interactive proof for verifying distribution property $\mathcal{P}$ with sample complexity $S$ and communication complexity $C$.
\end{corollary}

    On the other hand, if we allow $\varepsilon = \Theta(\log s)$, then the protocol from \cite{Herman23_doubly_efficient} satisfies $\varepsilon$-local DP in the communication phase and allows us to achieve $\tilde{O}(\sqrt{N})$ sample and communication complexity for verifying label invariant properties, when the best known bounds for the public coin case is $\tilde{O}(N^{2/3})$. Thus, the bound on $\varepsilon$ obtained here is tight up to log factors if we require local DP in the communication phase.

\section{Independence Testing}
\label{sec:independence}
In this section we derive upper and lower bounds for one-message (MA) protocols for testing whether a distribution is a product distribution. That is, if our distribution $D$ is over the space $A_1 \times A_2 \times \dots \times A_d$, we wish to test, with the aid of a prover, whether $D$ is a product distribution, or $D$ is at least $\sigma$-far in total variation distance from any product distribution. Our results show that when we have access to a prover, testing independence has sample complexity $\tilde{\Theta}(\sqrt{N})$ where $N$ is the domain size and polynomial runtime. 

In our upper bound, we use the prover to reduce the problem of privately verifying that the distribution is a product distribution to private identity testing.  Recall that we are interested in the case in which the sample complexity is strictly smaller than the domain size.  In this realm, our protocol matches the bounds on sample complexity for unassisted private identity testing given in~\cite{ASZ18}, including the dependence on the closeness and privacy parameters (which in turn match the upper bounds in Algorithm~1 of~\cite{aliakbarpour2017differentiallyprivateidentitycloseness}).  We cannot directly use the existing private algorithms in~\cite{aliakbarpour2017differentiallyprivateidentitycloseness} because they are computationally inefficient, requiring time that scales with $N$ rather than $\sqrt{N}$. The inefficiency arises because of a specific computation in Goldreich's reduction from identity testing to uniformity testing~\cite{goldreich2011testing}, all for the non-private case.  Herman and Rothblum replace this expensive exact computation with an estimate. 
We were unable to adapt Herman and Rothblum's computationally efficient non-private algorithm~\cite{herman2024verifyreasonabledistributionproperty} as this relies on $\chi^2$-type test that has high sensitivity.  
Finally, we were unable to use this estimation trick to improve the computational complexity of the private uniformity test in~\cite{diakonikolas2018sample} as the latter relies on an intolerant total variation test. In contrast, our introduction of an estimate seems to require a tolerant test, which in general cannot be achieved in the parameter realm of interest to us, given the tolerance requirement imposed by the estimate; see~\cite{canonne2022price}.  A similar issue arises when trying to use the algorithm in~\cite{ASZ18}.
Fortunately, for the $o(N)$ sample size of interest to us, Paninski's unique elements test~\cite{Paninski} used in Algorithm~1 of~\cite{aliakbarpour2017differentiallyprivateidentitycloseness} is just right, and the sample complexity of our efficient version matches that of their inefficient algorithm. 

Algorithm~1 of~\cite{aliakbarpour2017differentiallyprivateidentitycloseness} begins by applying Goldreich's reduction from identity testing on distributions of support size $N$ to uniformity testing on a distribution over $[6N]$ (reviewed in Section~\ref{sec:identity-to-uniformity} below). At a high level, this reduction splits bins of weight $q(i)$ into light-weight bins, specifically of weight $1/6N$, with any fractional remainder relegated to a special residual bin.  In the inefficient step, the verifier computes the total weight relegated to the residual bin, and then splits this residual bin into light-weight bins.  The total number of lightweight bins will be $6N$.  When the inefficient step is replaced with an estimate, the tester no longer knows the exact weight in the residual bin (and thus it does not know how many light-weight bins into which it should be split).  Instead, we use an overestimate of the residual weight, split this final bin into a correspondingly larger number of light-weight bins, and pad the resulting split distribution with 0-probability labels so that it can be viewed as an explicit distribution over $[M]$ for some $M\ge 6N$.  We then run the differentially private estimator of Algorithm~1 in~ADR18 with respect to the reference distribution $U_M$.  That algorithm uses Paninski's test for uniformity over $M$ elements, which checks that the number of unique elements in the sample is sufficiently large as a function of $M$~\cite{Paninski}. The key point is that for $M$ sufficiently close to $6N$, the expected number of unique elements does not change much between $U_M$ and and $U_{6N}$.  Indeed, it is this same robustness that accommodates the noise introduced for privacy.

\subsection{Upper Bound}

\begin{theorem}
Define
\[
  \Lambda(N,\sigma,\eps)
  =\frac{\sqrt N}{\sigma\sqrt\eps}
   +\frac{\sqrt N}{\sigma^2}.
\]
There exists a universal constant \(c_*>0\) such that the following holds.  Let
\(\eps\in(0,1]\), \(\sigma\in(0,1/12)\), and suppose
\begin{equation}
  \Lambda(N,\sigma,\eps)\le c_*N.
  \label{eq:sparse-regime}
\end{equation}
There is an \(\eps\)-differentially private one-message proof protocol for verifying that
\(D\) is a product distribution with the following guarantees:
\begin{description}
  \item[Completeness.] If $D$ is a product distribution and the prover sends the exact
  marginals of \(D\), then the verifier accepts with probability at least \(2/3\).

  \item[Soundness.] If $D$ is $\sigma$-far from any product distribution over $A_1 \times A_2 \times \dots \times A_d$ in TV distance, then for every prover message the verifier accepts
  with probability at most \(1/3\).
\end{description}
The verifier uses
$s=O(\frac{\sqrt N}{\sigma\sqrt\eps}+\frac{\sqrt N}{\sigma^2}+\frac1{\sigma^2\eps})$ 
samples from \(D\). 
After
\(O(m)\) preprocessing of the prover's message, its expected verifier running time is $O\left(m+ds+\sigma^{-4}\log N\right).$
The communication cost from the prover is $O(\sum_i|A_i|)$. 
\end{theorem}
\begin{remark}
    This sample complexity is $O\left(\frac{\sqrt N}{\sigma\sqrt\eps}+\frac{\sqrt N}{\sigma^2}\right)$ when $s \le N$, which is the case of interest to us. This matches the lower bounds we will prove in theorem~\ref{thm:indep_lb} up to polylog$(N)$ factors. This matches the sample complexity of the non-private non-interactive algorithm for independence in \cite{ADK15} when $\varepsilon \ge \Omega(\sigma^2)$, and improves upon the sample complexity of $O\left(\frac{\sqrt{N}}{\sigma^2\varepsilon}\right)$ that one would get from applying the generic reduction to \cite{ADK15}.
\end{remark}
\subsubsection{The protocol}

The verifier has i.i.d. sample access to an unknown distribution \(D\) on \(\cX\).
Before the verifier draws its private sample, the prover sends one message consisting of
purported marginal distributions
\[
  M^{(i)}=(M^{(i)}_a)_{a\in A_i},\qquad i\in[d].
\]
The verifier first checks that every entry is nonnegative and that
\(\sum_{a\in A_i}M^{(i)}_a=1\) for each \(i\).  It rejects if a check fails.  Otherwise the
message defines the product distribution
\[
  Q(x_1,\ldots,x_d)=\prod_{i=1}^d M^{(i)}_{x_i}.
\]

After validating the prover's message, let \(Q\) be the induced product distribution and
form the quantities in Section~\ref{sec:identity-to-uniformity} with \(P=D\).  Fix
\(
  \eta=c_0\sigma^2,
\)
where \(c_0>0\) is a sufficiently small universal constant.

\begin{enumerate}[leftmargin=*,label=\textbf{Step \arabic*.}]
  \item Let $U_N$ be uniform on $[N]$. Draw
  \(t_Q=\lceil C_Q\eta^{-2}\rceil\) samples from \(q_1=(Q+U_N)/2\), which can be simulated from
  the explicit product distribution \(Q\).  For
  \(
    h(x)=1-\theta_x\in[0,1],
  \)
  set
  \(
    \widetilde b=\frac1{t_Q}\sum_{j=1}^{t_Q}h(Y_j).
  \)
  Letting \(b=\mathbb E_{q_1}h(Y)\), Hoeffding's inequality gives
  \(|\widetilde b-b|\le\eta\) except with a prescribed small constant probability.

  \item Draw \(t_D\) fresh samples from \(D\) and independently mix each with a uniform
  point to obtain samples \(X_j\sim p_1=(D+U_N)/2\).  Set
  \[
    \widetilde a
    =\frac1{t_D}\sum_{j=1}^{t_D}h(X_j)
     +\Lap\left(\frac{2}{t_D\eps}\right),
  \]
  where
  \(
    t_D=\left\lceil C_D\left(\eta^{-2}+\frac1{\eta\eps}\right)\right\rceil.
  \)
  The empirical mean has sensitivity \(1/t_D\), so this step is \(\eps/2\)-DP.  By
  Hoeffding's inequality and the Laplace tail bound,
  \(|\widetilde a-a|\le2\eta\) except with a prescribed small constant probability.

  \item If \(|\widetilde a-\widetilde b|\ge4\eta\), reject.

  \item Otherwise define
  \[
    \overline m=
    \min\!\left\{6N,\,
      \max\!\left\{1,\left\lceil6N(\widetilde b+\eta)\right\rceil\right\}
    \right\},
    \qquad
    \overline b=\frac{\overline m}{6N},
    \qquad
    M=6N+\lceil12N\eta\rceil+2.
  \]
  From each fresh sample from \(D\), generate one sample from an approximate split
  distribution \(\widehat R\): first simulate \(Y\sim p_2\); if $Y\in[N]$,
  output \((x,J)\) for uniform \(J\in[m_x]\), and if \(Y=\botel\), output
  \((\botel,J)\) for uniform \(J\in[\overline m]\).  Count equality of these pair labels;
  the algorithm never needs to enumerate the alphabet.  On the good estimation event (that is, when $\{|\widetilde b-b|\le\eta\}\cap
\{|\widetilde a-a|\le2\eta\}$), the
  active labels are embedded into an alphabet of size \(M\), with all remaining labels given
  zero probability.  The clipping in the definition of \(\overline m\) makes the protocol
  well-defined even off the good estimation event; it is inactive on the good event except
  that it assigns one unused label when \(a=b=0\).

  \item Run the test in Lemma~\ref{lem:robust-unif} on \(\widehat R\) with
  \(
    \alpha=\sigma/4, \xi=\eps/2,
  \)
  and return its decision.
\end{enumerate}

\subsubsection{A robust private unique-elements test}

For a distribution \(R\) on an alphabet of size \(M\), let \(K\) be the number of symbols
that occur exactly once in \(s\) i.i.d. draws from \(R\), and let
\[
  \mu_M=s\left(1-\frac1M\right)^{s-1}
\]
be its expectation under \(U_M=\Unif([M])\).

We record the standard singleton inequalities in the form needed below.  They are the
analytic core of the private unique-elements uniformity tester of ADR18.

\begin{lemma}[Singleton inequalities]\label{lem:singleton-ineq}
There exist universal constants \({c_{\rm gap},C_{\rm var}>0}\) such that, whenever
\({2\le s<M}\), every distribution \(R\) on an alphabet of size \(M\) satisfies
\begin{align}
  0&\le \mu_M-\mathbb E_R K,
  \label{eq:uniform-max}\\
  \mu_M-\mathbb E_R K
  &\ge c_{\rm gap}\frac{s^2}{M}\TV(R,U_M)^2,
  \label{eq:singleton-gap}\\
  \operatorname{Var}_R(K)
  &\le C_{\rm var}\left(\mu_M-\mathbb E_R K+\frac{s^2}{M}\right).
  \label{eq:singleton-var}
\end{align}
The constants can be chosen independently of \(M,s\), and \(R\).
\end{lemma}

\begin{remark}
    Paninski \cite{Paninski} states that he assumes (after converting to the notation directly above) $s = o(M)$.  With a slight modification of Paninski's argument~\cite{Paninski}, the  bounds hold for all $2\le s<M$. We provide a full proof of the strengthened result for the mean bound in Appendix~\ref{app:Paninski}. The variance bound follows from Paninski's Efron--Stein argument~\cite{Paninski}.
\end{remark}

The following lemma gives a more general version of Algorithm 1 from~\cite{aliakbarpour2017differentiallyprivateidentitycloseness}.
\begin{lemma}[Generalized Private uniformity test]\label{lem:robust-unif}
{There are universal constants \(C_{\rm UE}>0\) and \(c_{\rm sp}\in(0,1)\)} with the following property.  Let
\(\alpha,\xi\in(0,1]\), and choose
\[
  s\ge C_{\rm UE}\left(
      \frac{\sqrt M}{\alpha^2}
      +\frac{\sqrt M}{\alpha\sqrt\xi}
      \right),
  \qquad s\le c_{\rm sp}M.
\]
Set
\[
  L=\frac{s^2\alpha^2}{M},
  \qquad
  T=\mu_M-\frac{c_{\rm gap}}2L.
\]
The test that releases
\[
  Z=K+\Lap(2/\xi)
\]
and accepts exactly when \(Z\ge T\) is \(\xi\)-differentially private.  Moreover:
\begin{enumerate}[label=(\roman*)]
  \item if
  \[
    0\le \mu_M-\mathbb E_RK
    \le \frac{c_{\rm gap}}4L,
    \label{eq:robust-complete}
  \]
  then it accepts with probability at least \(5/6\);
  \item if \(\TV(R,U_M)\ge\alpha\), then it rejects with probability at least \(5/6\).
\end{enumerate}
\end{lemma}
\subsubsection{The identity-to-uniformity splitting map}\label{sec:identity-to-uniformity}
We review here Goldreich's reduction from identity testing to uniformity testing~\cite{goldreich2020uniform}.
Identify \(\cX\) with \([N]\), and let \(U_N\) be uniform on \(\cX\).  For arbitrary
distributions \(P,Q\) on \(\cX\), define
\[
  p_1=\frac12(P+U_N),
  \qquad
  q_1=\frac12(Q+U_N).
\]
For \(x\in\cX\), put
\[
  m_x=\lfloor 6Nq_1(x)\rfloor,
  \qquad
  \theta_x=\frac{m_x}{6Nq_1(x)}.
\]
Because \(q_1(x)\ge1/(2N)\), we have \(m_x\ge3\) and
\begin{equation}
  \theta_x\ge\frac23.
  \label{eq:theta-lower}
\end{equation}
Define distributions \(p_2,q_2\) on \(\cX\cup\{\botel\}\) by
\[
  p_2(x)=\theta_xp_1(x),\qquad
  q_2(x)=\theta_xq_1(x)=\frac{m_x}{6N},
\]
and put all remaining mass at \(\botel\).  Write
\[
  a=p_2(\botel),
  \qquad b=q_2(\botel).
\]
In particular, \(6Nb\) is a nonnegative integer.

The exact splitting map draws \(Y\sim p_2\), then draws
\(J\) uniformly from \([6Nq_2(Y)]\), and outputs \((Y,J)\).  If \(b=0\), then every
\(6Nq_1(x)\) is an integer, hence every \(\theta_x=1\) and consequently \(a=0\); thus the
zero-size \(\botel\)-block is never sampled.  Denote the output distribution by
\(R_0=F_Q(P)\).  Its alphabet has exactly
\(\sum_y6Nq_2(y)=6N\) elements.

\begin{lemma}[Splitting-map guarantees]\label{lem:splitting}
The exact splitting map satisfies
\begin{enumerate}[label=(\roman*)]
  \item if \(P=Q\), then \(R_0=U_{6N}\);
  \item for all \(P,Q\),
  \[
    \TV(R_0,U_{6N})\ge\frac13\TV(P,Q).
  \]
\end{enumerate}
\end{lemma}

\subsection{Proof of Theorem 5.1}

Choose the constants \(C_Q,C_D\) so that the event
\begin{equation}
  \mathcal G=
  \left\{|\widetilde b-b|\le\eta
  \ \text{and}\ 
  |\widetilde a-a|\le2\eta\right\}
  \label{eq:good-event}
\end{equation}
has probability at least \(23/24\).  This follows directly from
\[
  \Pr\left(\left|\frac1t\sum_{j=1}^t Z_j-\mathbb EZ_j\right|>\eta\right)
  \le2e^{-2t\eta^2}
\]
for \(Z_j\in[0,1]\), and
\[
  \Pr\left(|\Lap(2/(t_D\eps))|>\eta\right)
  =e^{-t_D\eps\eta/2}.
\]

On \(\mathcal G\), the one-sided rounded estimate satisfies
\begin{equation}
  0\le\delta:=\overline b-b
  \le2\eta+\frac1{6N}.
  \label{eq:delta-bound}
\end{equation}
Moreover, if Step 3 does not reject, then
\begin{equation}
  |a-b|
  \le |a-\widetilde a|
      +|\widetilde a-\widetilde b|
      +|\widetilde b-b|
  <7\eta.
  \label{eq:a-b-bound}
\end{equation}

\subsubsection{Completeness}

Suppose \(D\) is a product distribution and the prover sends its exact marginals.  Then
\(Q=D\), hence \(p_1=q_1\), \(p_2=q_2\), and \(a=b\).  On \(\mathcal G\),
\[
  |\widetilde a-\widetilde b|\le3\eta<4\eta,
\]
so the verifier does not reject in Step 3.

Let \(K\) be the singleton count from \(s\) samples of \(\widehat R\).  The
\(6N(1-b)\) labels outside the residual block each have probability \(1/(6N)\), while the
\(\overline m\) residual labels each have probability \(b/\overline m\le1/(6N)\).  Therefore
\begin{align*}
  \mathbb E_{\widehat R}K
  &=s(1-b)\left(1-\frac1{6N}\right)^{s-1}
    +sb\left(1-\frac b{\overline m}\right)^{s-1}\\
  &\ge s\left(1-\frac1{6N}\right)^{s-1}
   =:\mu_{6N}.
\end{align*}
Thus enlarging the residual block cannot decrease the expected number of singletons.

Let \(f_s(k)=s(1-1/k)^{s-1}\).  Since
\[
  0\le f_s'(k)
  =\frac{s(s-1)}{k^2}\left(1-\frac1k\right)^{s-2}
  \le\frac{s^2}{k^2},
\]
the mean-value theorem and the definition of \(M\) give
\begin{align}
  0
  \le \mu_M-\mathbb E_{\widehat R}K
  &\le f_s(M)-f_s(6N)\notag\\
  &\le\frac{s^2(M-6N)}{(6N)^2}\notag\\
  &\le\frac{s^2}{N}\left(\frac\eta3+\frac1{12N}\right).
  \label{eq:complete-centering}
\end{align}
Here \(\mu_M-\mathbb E_{\widehat R}K\ge0\) also follows from
Lemma~\ref{lem:singleton-ineq} once \(s\le c_{\rm sp}M\).

Now set \(\alpha=\sigma/4\).  By first choosing \(c_0\) in
\(\eta=c_0\sigma^2\) sufficiently small and then choosing the universal constant \(c_*\)
in \eqref{eq:sparse-regime} sufficiently small, the right-hand side of
\eqref{eq:complete-centering} is at most
\[
  \frac{c_{\rm gap}}4\frac{s^2\alpha^2}{M}.
\]
The role of \eqref{eq:sparse-regime} here is twofold: it guarantees
\(s\le c_{\rm sp}M\) and makes the rounding term \(1/N\) negligible compared with
\(\sigma^2\).  Lemma~\ref{lem:robust-unif} therefore gives conditional acceptance
probability at least \(5/6\).  Consequently
\[
  \Pr(\text{accept})
  \ge1-\Pr(\mathcal G^c)-\frac16
  \ge1-\frac1{24}-\frac16
  =\frac{19}{24}>\frac23.
\]

\subsubsection{Soundness}

Suppose
\(\TV(D,\mathsf{Prod}(\cX))\ge\sigma\).  For every valid prover message, the induced
\(Q\) is a product distribution, so
\begin{equation}
  \TV(D,Q)\ge\sigma.
  \label{eq:D-Q-far}
\end{equation}
By Lemma~\ref{lem:splitting}, the exact split distribution satisfies
\begin{equation}
  \TV(R_0,U_{6N})\ge\sigma/3.
  \label{eq:exact-far}
\end{equation}

Condition on \(\mathcal G\).  If Step 3 rejects, soundness already holds.  Otherwise
\eqref{eq:a-b-bound} holds.  The exact and approximate split distributions differ only in
the residual block.  If \(b>0\), then
\begin{align}
  \TV(R_0,\widehat R)
  &=a\,\TV(U_{6Nb},U_{6N\overline b})\notag\\
  &=a\frac{\overline b-b}{\overline b}\notag\\
  &\le \delta+|a-b|\notag\\
  &\le9\eta+\frac1{6N}.
  \label{eq:split-perturbation}
\end{align}
If \(b=0\), then \(a=0\), as observed in Section~2.2, so
\eqref{eq:split-perturbation} remains valid.

On \(\mathcal G\), the number of active labels of \(\widehat R\) is
\[
  6N(1-b)+6N\overline b=6N(1+\delta)\le M.
\]
Embed both \(R_0\) and \(\widehat R\) into the same \(M\)-element alphabet.  Since
\[
  \TV(U_{6N},U_M)=1-\frac{6N}{M}
  \le2\eta+\frac1{2N},
\]
the triangle inequality, \eqref{eq:exact-far}, and
\eqref{eq:split-perturbation} yield
\begin{align}
  \TV(\widehat R,U_M)
  &\ge\TV(R_0,U_{6N})
       -\TV(R_0,\widehat R)
       -\TV(U_{6N},U_M)\notag\\
  &\ge\frac\sigma3-11\eta-\frac{2}{3N}.
  \label{eq:key-soundness}
\end{align}
Choose \(c_0\) and \(c_*\) sufficiently small.  Since
\(\sigma<1/12\) and \eqref{eq:sparse-regime} makes \(1/N=o(\sigma)\) uniformly in the
allowed parameter range, \eqref{eq:key-soundness} is at least \(\sigma/4\).
Lemma~\ref{lem:robust-unif} therefore gives conditional rejection probability at least
\(5/6\).  Hence, for every prover,
\[
  \Pr(\text{accept})
  \le\Pr(\mathcal G^c)+\frac16
  \le\frac1{24}+\frac16
  =\frac5{24}<\frac13.
\]

\subsubsection{Privacy}

The prover's message and the estimate \(\widetilde b\) use no private verifier sample.
The release \(\widetilde a\) is \(\eps/2\)-DP because its sensitivity is \(1/t_D\) and its
Laplace scale is \(2/(t_D\eps)\).  The map taking one fresh sample from \(D\) to one sample
from \(\widehat R\) depends only on \(Q\), \(\widetilde b\), and public randomness; thus
neighboring input datasets map to datasets differing in at most one observation.  The
unique-elements release is \(\eps/2\)-DP by Lemma~\ref{lem:robust-unif}.  Post-processing
and basic composition show that the entire transcript is \(\eps\)-DP.

\subsubsection{Sample, running-time, and communication bounds}

Since \(\eta=c_0\sigma^2\), the private residual-mass estimate uses
\[
  t_D=O\left(\sigma^{-4}+\frac1{\sigma^2\eps}\right)
\]
samples from \(D\).  Lemma~\ref{lem:robust-unif}, with
\(M=\Theta(N)\), \(\alpha=\sigma/4\), and \(\xi=\eps/2\), uses
\[
  s=O\left(
      \frac{\sqrt N}{\sigma^2}
      +\frac{\sqrt N}{\sigma\sqrt\eps}
      \right)
\]
additional samples.    Furthermore,
\eqref{eq:sparse-regime} implies
\[
  \sigma^{-4}
  =\frac1N\left(\frac{\sqrt N}{\sigma^2}\right)^2
  =O\left(\frac{\sqrt N}{\sigma^2}\right).
\]

Reading, validating, and preprocessing the marginals costs \(O(m)\) real-RAM operations.
Evaluating \(Q(x)\), and hence \(q_1(x),m_x,\theta_x\), costs \(O(d)\) per verifier
sample.  After the marginal cumulative distribution functions have been constructed,
sampling from \(Q\) takes \(O(\sum_i\log|A_i|)=O(\log N)\) time.  The
\(O(\sigma^{-4})\) simulated samples from \(q_1\) therefore cost
\(O(\sigma^{-4}\log N)\).  A hash table computes the singleton count in expected linear
time in the sample size. 

\subsection{Lower Bound}

We provide a reduction from private uniformity verification to private independence verification (both with a prover).  The reduction has minimal impact on sample complexity and privacy.  We then observe that (private) uniformity testing reduces to (private) uniformity verification, again with no impact on the parameters, so lower bounds on sample complexity for uniformity testing yield lower bounds on independence verification (\cite{ASZ18}). Our algorithm matches these bounds.

Interactive proofs cannot improve upon the sample complexity of uniformity testing, as noted by \cite{herman2022verifying}. The proof of this is fairly short so we include it here for the convenience of the reader.

\begin{proof}
    Given an interactive proof for uniformity over $[N]$, the tester can exactly replicate what the verifier would do to produce its messages (possibly given a response from the prover) as follows.  Whenever the prover is supposed to sample from the distribution $D$ (that it is claiming is uniform), the tester samples from the uniform distribution on $[N]$. Other than this, it exactly simulates the prover's computations at every round (This might be computationally inefficient).
    
    We now show that this protocol satisfies completeness and soundness:

    \textbf{Completenesss:} If the distribution is indeed uniform, this interaction is exactly the same as what would happen if the prover were honest, so completeness is immediate.

    \textbf{Soundness:} If the distribution is far from uniform, the soundness guarantee tells us that no matter what the prover does we will reject with high probability, including if the prover behaves as described above.
    
\end{proof}    

    Thus we can simulate the interactive proof with the same sample complexity. Given that the known algorithm for uniformity with matching sample complexity is already equal to the sample complexity, this says there's nothing that can possibly be improved on for uniformity with a prover. As such, even when we have access to a prover, reduction to uniformity results in the same lower bound. This holds even when we introduce privacy, as the same argument holds and the privacy loss of the algorithm is less than the privacy loss of the interactive proof because it doesn't need to publicize the messages.

\begin{theorem}
\label{thm:indep_lb}
    Suppose $A$ is an $\varepsilon$-DP protocol with $\varepsilon \ge \Omega\left(\frac{1}{s^{0.49}}\right)$ for verifying independence with $s$ samples from a distribution $D$ over the space $A_1 \times A_2 \times \dots \times A_d$ for sufficiently large $d$ and $C$ communication complexity. Further, suppose $A$ has the following completeness and soundness guarantees:

\textbf{Completeness}: If $D$ a product distribution and the prover is honest, then the verifier will accept with probability at least $0.9$.

\textbf{Soundness}: If $D$ is at least $\sigma$-far from a product distribution, then the verifier will accept with probability at most $0.1$ for any prover.

Then, there exists a $2\varepsilon$-DP protocol for verifying uniformity with sample complexity at most $s$ and communication complexity $C$ with the following completeness and soundness guarantees: 

\textbf{Completeness:} If $D$ is uniform and the prover is honest, then the verifier accepts with probability at least $0.75$.

\textbf{Soundness:}  If $D$ is at least $\sigma\left(1+\frac{3\sqrt{\log N}}{2}\right)+\frac{5\log N\sqrt{\log \log N}}{\sqrt{s}}$-far from uniform, then the verifier rejects with probability at least $0.75$ for any prover.
\end{theorem}
\begin{remark}
\label{remark: meaningful}
The theorem gives meaningful bounds when $\sigma\left(1+\frac{3\sqrt{\log N}}{2}\right)+\frac{5\log N\sqrt{\log \log N}}{\sqrt{s}}<1$.
\end{remark}

\begin{proof}[Proof of Theorem~\ref{thm:indep_lb}]
Let $d$ be the smallest integer such that $2^d \ge N$. Define the distribution
$$D' = \frac{N}{2^d}D+\frac{2^d-N}{2^d}\Unif[N+1,2^d].$$
The following procedure allows us to use sample access from $D$ to obtain sample access to $D'$:\\

With probability $\frac{N}{2^d}$, draw a sample from $D$. Otherwise, uniformly generate an integer from the range $[N+1,2^d]$. By definition,
$$\frac{N}{2^d}D_{TV}(D,\Unif[N])=D_{TV}(D',\Unif[2^d])$$
so verifying uniformity of $D$ is equivalent (up to a constant factor) to verifying uniformity of $D'$.
        
We construct a protocol for verifying uniformity of some distribution $D'$ over the domain $[2^d]$ as follows:
\begin{enumerate}
    \item Let $D''$ be a distribution over the space $A_1 \times A_2 \times \dots \times A_d$ where $A_i= \{0,1\}$ for each $i$ defined by $D'' = f(D')$ where $f:[2^d] \rightarrow \{0,1\}^d$ is any bijection.
    \item Prover and Verifier run $A$ to determine whether $D''$ is a product distribution. If the verifier rejects in the product verification protocol $A$, then it rejects.
    \item If the verifier accepted in the product verification protocol $A$,
 let the empirical marginals be denoted $\hat{p}_i$. If $\max\left|\hat{p}_i - \frac{1}{2}\right|+ \Lap\left(\frac{1}{\varepsilon s}\right) $ lies outside the range $[-\frac{7\sqrt{\ln d}}{2\sqrt{s}},\frac{7\sqrt{\ln d}}{2\sqrt{s}}]$, the verifier rejects.\footnote{$\max_{i\in d}|\hat{p}_i-\frac{1}{2}|$ has sensitivity at most $1/s$, so it suffices to add noise scaled to $\frac{1}{\varepsilon s}$.
 }
    \item If the verifier has not rejected, it accepts.
    \end{enumerate}

\textbf{Completeness:} If the prover is honest and $D'$ is uniform, then $D''$ is a product distribution over $\{0,1\}^d$ so the product verification protocol in Step 2 accepts with probability $0.9$. Then, the empirical marginals are binomial distributions and we have the following Chernoff bound, where $X$ is the dataset and $X_i$ is the count of the number of samples where $A_i = 1$:

$$\Pr_{X \sim (D'')^s}\left[\left|X_i-\frac{s}{2}\right| \ge \frac{\delta s}{2}\right] \le 2e^{-\delta^2s/6}$$
For sufficiently large $d$, we get $2e^{-\delta^2s/6} \le \frac{0.14}{d}$ when $\delta \ge \sqrt{\frac{6\ln d - 6 \ln(0.07)}{s}}$. This is because:

\begin{align*}2e^{-\delta^2s/6} &\le \frac{0.14}{d}\\
e^{-\delta^2s/6} &\le \frac{0.07}{d}\\
-\frac{\delta^2s}{6} &\le \ln(0.07) - \ln d\\
\delta &\ge \sqrt{\frac{6\ln d - 6 \ln(0.07)}{s}}\\
\end{align*}

Since we assumed $\varepsilon \ge \Omega\left(\frac{1}{s^{0.49}}\right)$, we have

$$\Pr\left[\left|\Lap\left(\frac{1}{\varepsilon s}\right)\right| > \frac{1}{\sqrt{s}}\right] = O\left(e^{-\varepsilon\sqrt{s}}\right) < 0.01$$

Now, consider $d$ large enough that $(25/4)\ln d-7\sqrt{\ln d} + 1 + 6 \ln(0.07) > 0$. This ensures that when we let $\delta = \sqrt{\frac{6\ln d - 6 \ln(0.07)}{s}}$, we have
$$\delta + \frac{1}{\sqrt{s}} < \frac{7\sqrt{\ln d}}{2\sqrt{s}}$$

Plugging $\delta = \sqrt{\frac{6\ln d - 6 \ln(0.07)}{s}}$ into the Chernoff bound above, and taking the union bound over the $d$ values of $i$ for randomness from sampling and the Laplace noise being added once, we get:

$$\Pr_{X \sim (D'')^s, \eta_i\sim\Lap(1/\varepsilon s)}\left[X_i + \eta_i \notin \left[\frac{1}{2}-\frac{7\sqrt{\ln d}}{2\sqrt{s}}, \frac{1}{2}+\frac{7\sqrt{\ln d}}{2\sqrt{s}}\right]\forall i \in[d] \right] \le 0.15$$

so by a union bound, we accept in both steps $2$ and $3$ with probability at least $0.75$.

\textbf{Soundness:}
Suppose
$$D_{TV}(D'', \Unif\{0,1\}^dt) \ge \sigma\left(1+\frac{3\sqrt{d}}{2}\right)+\frac{5d\sqrt{\log d}}{\sqrt{s}}$$
and let $D^{(0)}$ be the product distribution with the same marginals as $D''$. By the triangle inequality, we must either have
$$D_{TV}(D^{(0)}, \Unif\{0,1\}^d) \ge\frac{5d\sqrt{\log d}}{\sqrt{s}}$$
or
$$D_{TV}(D'',D^{(0)}) \ge \sigma\left(1+\frac{3\sqrt{d}}{2}\right).$$

In the former case, letting $D_i$ denote the marginal probability that the $i^{th}$ coordinate of $D^{(0)}$ equals $1$, we have

$$\frac{5d\sqrt{\log d}}{\sqrt{s}} \le D_{TV}(D^{(0)}, \Unif\{0,1\}^d) \le \sum_{i=1}^d \left|D_i - \frac{1}{2}\right|$$
so there must exist some $i$ such that
$$\left|D_i-\frac{1}{2}\right| \ge \frac{5\sqrt{\log d}}{\sqrt{s}}.$$
WLOG, suppose $D_i > \frac{1}{2}$. By Hoeffding's inequality, we have
$$\Pr_{X_1, X_2, \dots, X_s \sim D^{(0)}_i}\left[\sum_{j=1}^s X_j \le \frac{s}{2}+\frac{7}{2}\sqrt{s \log d} \right] \le e^{-\log d/2} = \frac{1}{\sqrt{d}}$$
and just like in the proof of completeness, we have
$$\Pr\left[\left|\Lap\left(\frac{1}{\varepsilon s}\right)\right| > \frac{1}{\sqrt{s}}\right] = O\left(e^{-\varepsilon\sqrt{s}}\right) < 0.01$$
and $\frac{\sqrt{\ln d}}{\sqrt{s}} > \frac{1}{\sqrt{s}}$ so the probability that the verifier accepts in step $3$ is at most $\frac{1}{\sqrt{d}}$ + 0.01.

In the latter case, we claim that if
$$Q^* = \min_{\text{product distributions }Q}D_{TV}(D'',Q),$$
then $D_{TV}(D'',D^{(0)}) \le \left(1+\frac{3\sqrt{d}}{2}\right) \cdot D_{TV}(D'',Q^*)$.

To see this, first note that by the triangle inequality,
$$D_{TV}(D'',D^{(0)}) \le D_{TV}(D'',Q^*) + D_{TV}(Q^*,D^{(0)}).$$

By Pinsker's inequality, we have

$$ D_{TV}(Q^*,D^{(0)}) \le \sqrt{\frac{1}{2}KL(Q^*,D^{(0)})}$$
and since $Q^*$ and $D^{(0)}$ are product distributions, if we let $Q^*_i$ denote the probability distribution of the $i^{th}$ coordinate of $Q^*$ and similarly let $D^{(0)}_i$ denote the probability distribution of the $i^{th}$ coordinate of $D^{(0)}$, we have

$$KL(Q^*||D^{(0)}) = \sum_{i=1}^d KL(Q^*_i||D^{(0)}_i).$$

Now, suppose $D^{(0)}_i = \operatorname{Bern}(x_i)$ and $Q^*_i = \operatorname{Bern}(y_i)$.

$$D_{KL}(Q^*_i||D^{(0)}_i) = y_i \log \frac{y_i}{x_i} + (1-y_i) \log\frac{1-y_i}{1-x_i} \le y_i\frac{y_i-x_i}{x_i} + (1-y_i)\frac{x_i-y_i}{1-x_i} = (y_i-x_i) \left(\frac{y_i}{x_i}-\frac{1-y_i}{1-x_i}\right)$$

where the last inequality follows from applying the identity $\log t \le t-1$. Additionally, note that $|x_i - y_i| = D_{TV}(D^{(0)}_i,Q^*_i) \le D_{TV}(D'',Q^*)$ since projecting onto the $i^{th}$ coordinate can only decrease the distance. We now note that as we just showed above, if the protocol does not reject in step $3$, then with high probability over both sampling and noise added for privacy, we have $x_i \in \left[\frac{1}{3},\frac{2}{3}\right]$. We now condition on this event for the following computations. Since $y_i \in [0,1]$, this means
$$ (y_i-x_i)\left(\frac{y_i}{x_i}-\frac{1-y_i}{1-x_i}\right) = \left(\frac{(y_i-x_i)^2}{x_i(1-x_i)}\right) \le \frac{9}{2}(y_i-x_i)^2.$$
Putting this all together, we get:
$$D_{KL}(Q^*||D^{(0)}) \le \sum_{i=1}^d\frac{9}{2}\left(D_{TV}(D^{(0)}_i,Q^*_i)\right)^2 \le \frac{9d}{2}\left(D_{TV}(D'',Q^*)\right)^2$$

And plugging this into Pinsker's inequality:

$$D_{TV}(Q^*,D^{(0)}) \le \sqrt{\frac{9d}{4}\left(D_{TV}(D'',Q^*)\right)^2} = \frac{3\sqrt{d}}{2}D_{TV}(D'',Q^*)$$

And thus we have

$$D_{TV}(D'',D^{(0)}) \le \left(1+\frac{3\sqrt{d}}{2}\right)D_{TV}(D'',Q^*)$$

Thus, if $D_{TV}(D'',D^{(0)}) \ge \sigma\left(1+\frac{3\sqrt{d}}{2}\right)$, then $D_{TV}(D'',Q^*) \ge \sigma$ and soundness of $A$ tells us that the verifier rejects in Step 2 with probability at least $0.9$. Hence, taking a union bound, the probability that the verifier rejects overall is at least $0.88 > 0.75$.


\textbf{Privacy}

    We access the sample drawn from $D$, the distribution we wish to verify for uniformity, in step 2 and step 3. The distribution we're verifying for independence in step 2 is $D' = \frac{N}{2^d}D + \frac{2^d-N}{2^d} U$ where $U$ is uniform over the $[2^d] \setminus [N]$. Equivalently, each sample drawn for independence verification is iid Bernoulli as to whether it's drawn from $D$ or not.

    In particular, when we run the independence tester in step $2$, each sample drawn from $D$ appears in $D'$ at most once. Thus, since the independence tester satisfies $\varepsilon$-DP with respect to the sample drawn from $D'$, it also satisfies $\varepsilon$-DP with respect to the sample drawn from $D$.

    As noted above, $\max_{i\in d}|\hat{p}_i-\frac{1}{2}|$ has sensitivity at most $1/s$. Thus, in step $3$, we simply compute the maximum $|\hat{p}_i-\frac{1}{2}|$
    and add $\Lap(\frac{1}{\varepsilon s})$
    to ensure $\varepsilon$-DP.

    Thus, this is a protocol which satisfies $2\varepsilon$-DP for verifying uniformity.
\end{proof}


\begin{theorem}
\label{thm:unif_lb}
    (Section 5.3 in \cite{ASZ18}, restated) When $s \le N$, the sample complexity of testing uniformity with $\varepsilon$-DP, $\sigma$ soundness error, and domain $N$ is at least
    $$\Omega\left(\frac{\sqrt{N}}{\sigma^2} + \frac{\sqrt{N}}{\sigma \sqrt{\varepsilon}}\right)$$
\end{theorem}

\begin{corollary}
Let $\sigma\left(1+\frac{3\sqrt{\log N}}{2}\right)+\frac{5\log N\sqrt{\log \log N}}{\sqrt{s}}<1$.
    Suppose $A$ is a $\varepsilon$-DP protocol for verifying independence with $s$ samples from a distribution $D$ over the space $A_1 \times A_2 \times \dots \times A_d$, $\varepsilon \ge \Omega(1/s^{0.49})$, and $C$ communication complexity. Further, suppose $A$ has the following completeness and soundness guarantees:

\textbf{Completeness}: If $D$ a product distribution and the prover is honest, then the verifier will accept with probability at least $0.9$.

\textbf{Soundness}: If $D$ is at least $\sigma$-far from a product distribution, then the verifier will accept with probability at most $0.1$ for any prover.

Then $s \ge \tilde{\Omega}\left(\frac{\sqrt{N}}{\sigma^2} + \frac{\sqrt{N}}{\sigma \sqrt{\varepsilon}}\right)$
\end{corollary}

\begin{proof}
Theorem~\ref{thm:indep_lb} shows that a $\varepsilon$-DP protocol to verify independence with $s$ samples allows us to verify uniformity with $2\varepsilon$-DP, $s$ samples, and error threshold $\sigma\left(1+\frac{3\sqrt{\log N}}{2}\right)+\frac{5\log N\sqrt{\log \log N}}{\sqrt{s}}$. Theorem~\ref{thm:unif_lb} then tells us that
    \begin{align*}s &\ge \frac{\sqrt{N}}{\left(\sigma\left(1+\frac{3\sqrt{\log N}}{2}\right)+\frac{5\log N\sqrt{\log \log N}}{\sqrt{s}}\right)^2} + \frac{\sqrt{N}}{\left(\sigma\left(1+\frac{3\sqrt{\log N}}{2}\right)+\frac{5\log N\sqrt{\log \log N}}{\sqrt{s}}\right)\sqrt{2\varepsilon}}\\
    &\ge \frac{\sqrt{N}}{\left(\tilde{O}(\sigma) + \tilde{O}(1/\sqrt{s})\right)^2} + \frac{\sqrt{N}}{\left(\tilde{O}(\sigma) + \tilde{O}(1/\sqrt{s})\right) \cdot \sqrt{\varepsilon}}\\
    &\ge \tilde{\Omega}\left(\frac{\sqrt{N}}{\sigma^2} + \frac{\sqrt{N}}{\sigma \sqrt{\varepsilon}}\right)
    \end{align*}
    which matches our upper bound up to poly-logarithmic factors in $N$. The final inequality comes from noting that if $\tilde{O}(1/\sqrt{s})$ were the dominant term in the denominator(s), we would require
    $$s \ge \frac{\sqrt{N}}{\left(\tilde{O}(\sigma) + \tilde{O}(1/\sqrt{s})\right)^2} = \tilde{\Omega}\left(s \sqrt{N}\right)$$
    which is impossible.
\end{proof}

\section{Generic Reduction for \texorpdfstring{$\varepsilon$}{Epsilon}-DP AM Protocols}
\label{sec:generic}
Appendix B of \cite{aliakbarpour2017differentiallyprivateidentitycloseness,aliakbarpour2017differentiallyprivateidentityclosenessArxiv}\footnote{In this section, we will cite both the proceedings version and the full version of this work when citing results from this work. This particular result only appears in the full version, or \cite{aliakbarpour2017differentiallyprivateidentityclosenessArxiv}.} gives a generic reduction that converts any non-private tester to a private tester which increases the sample complexity by a factor of $O(1/\varepsilon)$:

\begin{algorithm}
    \caption{Generic Private Tester \cite{aliakbarpour2017differentiallyprivateidentitycloseness,aliakbarpour2017differentiallyprivateidentityclosenessArxiv}}
    \label{alg:generic}
    \begin{algorithmic}
        \Require Non-private distribution tester $A$ with error probability at most $\frac{1}{6}$, sample access to a distribution $D$, and privacy loss parameter $\varepsilon$
        \State $m = \left\lceil \frac{6}{\varepsilon}\right\rceil$
        \State $s'  = m \cdot s(N,\xi)$
        \State $x_1, x_2, \dots x_{s'} \sim D$
        \State $r \sim \Unif[m]$
        \State $O \leftarrow A\left(\{x_{(r-1)s+1}, x_{(r-1)s+2}, \dots, x_{rs}\}\right)$
        \State With probability $\frac{1}{6}$, flip the value of $O$
        \State Return $O$
    \end{algorithmic}
\end{algorithm}
\begin{lemma}
    (From \cite{aliakbarpour2017differentiallyprivateidentitycloseness,aliakbarpour2017differentiallyprivateidentityclosenessArxiv}). Suppose $\mathcal{A}$ is a tester that draws $s(N,\xi)$ samples from distribution $D$ and decides whether $D$ is within $\xi$ total variation distance of some property $\mathcal{P}$ with error probability at most $\frac{1}{6}$. Then, Algorithm~\ref{alg:generic} is $\varepsilon$-DP with sample complexity $O(s(N,\xi)/\varepsilon)$, and decides whether $D$ is within $\xi$ total variation distance of $\mathcal{P}$ with error probability at most $\frac{1}{3}.$
\end{lemma}
Now, we claim that this generic reduction can be extended to public coin interactive proofs.

\begin{theorem}
Suppose AM distribution tester $T$ can test whether some distribution $D$ $\xi$-approximately satisfies $\mathcal{P}$ for some property $\mathcal{P}$ with probability $\frac{5}{6}$. Suppose further that $T$ has sample complexity $s$ and communication complexity $c$. There exists an $\varepsilon$-DP AM tester $T'$ that can test whether $D$ $\xi$-approximately satisfies $\mathcal{P}$ with probability $\frac{2}{3}$ using sample complexity $\left\lceil\frac{6}{\varepsilon}\right\rceil s$ and communication complexity $c$.
\end{theorem}

\begin{proof}
We construct $T'$ as follows: 
\begin{enumerate}
    \item the honest prover in $T'$ is identical to that in $T$.
    \item In the $i^{th}$ round of interaction, when the verifier in $T$ would send random bits $\rho_i$ to the prover, the verifier in $T'$ does the same. All interactions are completed before doing anything else.

    \item The verifier in $T'$ draws $\left\lceil \frac{6}{\varepsilon}\right\rceil \cdot s$ samples from $D$, and then takes a random subsample of size $s$. Denote the subsample by $(x_1, x_2, \dots, x_s)$. $(x_1, x_2, \dots, x_s)$ is then used in place of the verifier's sample in $T$.
    
    \item The verifier in $T'$ does all calculations the verifier in $T$ would have done given the same transcript of interactions and sample $(x_1, x_2, \dots, x_s)$.

    \item Once simulation is complete, the verifier flips the answer with probability $\frac{1}{6}$ and outputs the result.
\end{enumerate}

\paragraph{Completeness and Soundness}
When $T'$ takes a sample of size $\left\lceil \frac{6}{\varepsilon}\right\rceil \cdot s$ from $D$ and then a random subsample of size $s$ from the sample, it is equivalent to taking a sample of size $s$ from $D$. Thus, the output of step 5 in $T'$ matches the output of $T$ with probability $\frac{5}{6}$ for any prover over the randomness of the verifier. Since $T$ is correct with probability $\frac{5}{6}$, $T'$ is correct with probability $\frac{25}{36}$.

\paragraph{Privacy}

In the public coin proof systems from Chiesa and Gur \cite{chiesa_et_al:LIPIcs.ITCS.2018.53}, the samples drawn are independent from the transcript of the protocol, and are drawn after the communication phase. Thus, the only output released that can contain information about the sample is the output of the protocol. Thus we can prove privacy identically to in the generic reduction above:

Let $X,Y$ be neighboring datasets which differ on the $i^{th}$ coordinate. Furthermore, let $i_1, i_2, \dots, i_s$ denote the indices of the subsample. Then, we have

\begin{align*}\Pr[T'(X) = \text{accept}] &= Pr[T'(X) = \text{accept} \cap i \notin \{i_1, i_2, \dots, i_s\}] + Pr[T'(X) = \text{accept} \cap i \in \{i_1, i_2, \dots, i_s\}]\\
&=Pr[T'(Y) = \text{accept} \cap i \notin \{i_1, i_2, \dots, i_s\}] + Pr[T'(X) = \text{accept} \cap i \in \{i_1, i_2, \dots, i_s\}]\\
&\le Pr[T'(Y) = \text{accept}] + Pr[i \in \{i_1, i_2, \dots, i_s\}]\\
&\le Pr[T'(Y) = \text{accept}] + \frac{\varepsilon}{6}
\end{align*}

And since both probabilities are at least $\frac{1}{6}$, we have

$$\frac{\Pr[T'(X) = \text{accept}]}{Pr[T'(Y) = \text{accept}]} \le 1 + \varepsilon \le e^\varepsilon$$

\end{proof}
\section{Improved Bounds for Label-Invariant Properties With Propose-Test-Release}
\label{sec:label_invariant}

In this section, we modify the AM protocol of~\cite{herman2024public} for verifying label-invariant distribution properties in order to protect the privacy of the verifier's sensitive data samples.  Since the original protocol is based on looking at various sorts of counts on a sample drawn from $D$, and since the original algorithm had to be robust to small changes in these counts because of randomness in the sample, it is intuitive that adding a little noise to these counts in order to protect privacy will not threaten completeness or soundness.  We prove that this is the case. Moreover, the verifier sample complexity of our privacy-preserving protocol will be roughly $\tO(N^{2/3}/\eps^{2/3})$, which is smaller than the cost of $O(N^{2/3}/\eps)$ that would be incurred by the generic private-to-non-private reduction presented in \Cref{sec:generic}.

The starting point of the \cite{herman2024public} protocol is the observation that label-invariant distribution properties depend only on the distribution's histogram of probability masses. With this in mind, their protocol proceeds by generating a uniform sample of domain elements and using the prover to determine the weights of each element in the sample. They call this process \emph{tagged sample retrieval}. They note that knowing these weights is sufficient to construct an approximate histogram and thus verify any label-invariant property. For reference, their tagged sample retrieval protocol is described below in \Cref{alg:non-private-label-invariant}, and its soundness and completeness guarantees are summarized in \Cref{thm:herman2024public}.

\begin{algorithm}
    \caption{Non-Private AM Protocol for Tagged Sample Retrieval From \cite{herman2024public}}
    \label{alg:non-private-label-invariant}
    \textbf{Require:} Proximity parameter $\sigma \in (0, 1)$, sample size $s = O(N^{2/3}\log(N)/\sigma^5)$, and ample access to a distribution $D$ over $[N]$ with $\max_{x \in [N]} D(x) \le p_{\mathrm{max}}$, where we define $p_{\mathrm{max}} = 1/s$ and $p_{\mathrm{min}} = \sigma/1000N$. We also define $\tau = \sigma^3/80000$.
    \begin{enumerate}
        \item $V$ draws uniform samples $S_1, \ldots, S_s \iid [N]$. $V$ rejects if any $x \in [N]$ appears over $\log N$ times among these samples. Otherwise, $V$ sends these samples to $P$.
        \item $P$ sends back tags $\pi_1, \ldots, \pi_s$ to $V$, where
        \[
            \pi_i = \begin{cases}
                D(S_i) &\text{if }D(S_i) \ge p_{\mathrm{min}},\\
                0 &\text{if }D(S_i) < p_{\mathrm{min}}.
            \end{cases}
        \]
        \item $V$ draws two fresh samples of size $s$ from $D$, denoted $T_1, \ldots, T_s \iid D$ and $T_1', \ldots, T_s' \iid D$.
        \item$V$ calculates, for each integer $j$ (possibly negative) such that $p_{\mathrm{min}} \le e^{j\tau}/N \le p_{\mathrm{max}}$, the bin
        \[
            S^j = \biggl\{i \in [s] \,\bigg|\, \pi_i \in \Bigl[\frac{e^{j\tau}}{N}, \frac{e^{(j+1)\tau}}{N}\Bigr)\biggr\}.
        \]
        $V$ calculates the number of two- and three-way collisions among $S, T, T'$ for each index $j$:
        \begin{align*}
            \pair_j &= \Bigl\lvert \bigl\{(k,r) \in [s]^2 \,\big|\, k \in S^j \text{ and }S_k = T_r\bigr\}\Bigr\rvert,\\
            \triple_j &= \Bigl\lvert\bigl\{(k,r,r') \in [s]^3 \,\big|\, k \in S^j\text{ and } S_k = T_r=T'_{r'}\bigr\}\Bigr\rvert.
        \end{align*}
        $V$ rejects if some bin with $\abs{S^j} \ge e^{-j\tau} \cdot s \cdot \sigma\tau/100\log(N)$ fails either of the following checks:
            \begin{align*}
                \pair_j &\in (1 \pm 4\tau) \cdot s \abs{S^j} \frac{e^{j\tau}}{N}, \\
                \triple_j &\in (1 \pm 4\tau) \cdot s^2\abs{S^j} \Bigl(\frac{e^{j\tau}}{N}\Bigr)^2.
            \end{align*}
        \item $V$ calculates $S^{-\infty} = \{i \in [s] \,|\,\pi_i < p_{\mathrm{min}}\}$ and the corresponding two-way collision count:
        \[
            \pair_{-\infty} = \Bigabs{\bigl\{(k,r) \in [s]^2\,\big|\,k \in S^{-\infty} \text{ and } S_k = T_r\bigr\}}.
        \]
        $V$ rejects if the following check fails:
        \[
            \pair_{-\infty}, \le s\abs{S^{-\infty}}\frac{\sigma}{50N}.
        \]
        
        \item $V$ accepts and outputs the tagged samples $(S_1, \pi_1), \ldots, (S_s, \pi_s)$.
    \end{enumerate}
\end{algorithm}

\begin{theorem}[\cite{herman2024public}]
\label{thm:herman2024public}
    Given a proximity parameter $\sigma > 0$ that is not too small (i.e. satisfying $1/\sigma \le N^{o(1)}$), along with access to $\tO(N^{2/3}) \cdot \mathrm{poly}(1/\sigma)$ samples from $D$, \Cref{alg:non-private-label-invariant} satisfies the following completeness and soundness conditions:
    \begin{itemize}
        \item \textbf{Completeness:} If the prover is honest, then with probability at least $3/4$, the verifier accepts and outputs tagged samples $(S_1, \pi_1), \ldots, (S_s, \pi_s)$ satisfying
        \[
            \frac{1}{s}\sum_{i \in [s] : \pi_i \ge \frac{\sigma}{1000N}} \biggl(1 - \min\Bigl\{\frac{\pi_i}{D(S_i)}, \frac{D(S_i)}{\pi_i}\Bigr\}\biggr) \le O(\sigma^3).
        \]
        and
        \[
            \frac{1}{s}\sum_{i\in[s]:\pi_i<\frac{\sigma}{1000N}} D(S_i) \le O\Bigl(\frac{\sigma}{N}\Bigr).
        \]
        \item \textbf{Soundness:} If the prover is dishonest, then with probability at least $3/4$, the verifier either rejects or outputs tagged samples $(S_1, \pi_1), \ldots, (S_s, \pi_s)$ satisfying
        \[
            \frac{1}{s}\sum_{i \in [s] : \pi_i \ge \frac{\sigma}{1000N}} \biggl(1 - \min\Bigl\{\frac{\pi_i}{D(S_i)}, \frac{D(S_i)}{\pi_i}\Bigr\}\biggr) \le O(\sigma).
        \]
        and
        \[
            \frac{1}{s}\sum_{i\in[s]:\pi_i<\frac{\sigma}{1000N}} D(S_i) \le O\Bigl(\frac{\sigma}{N}\Bigr).
        \]
    \end{itemize}
\end{theorem}

In this section, we shall demonstrate how to convert \Cref{alg:non-private-label-invariant} into a private protocol without paying the multiplicative factor of $O(1/\eps)$ that would be incurred by the generic private-to-non-private reduction presented in the previous section. Since the non-private protocol required roughly $N^{2/3}$ samples, the generic transformation would produce a private version using $N^{2/3}/\eps$ samples. Our algorithm, in contrast, will use a smaller sample size of $N^{2/3}/\eps^{2/3}$. Our privacy-preserving protocol is described in \Cref{alg:label-invariant}. For clarity, we have emphasized the key differences in blue. The algorithm's privacy and utility guarantees are summarized in \Cref{thm:label-invariant-privacy,thm:label-invariant-utility} below.

\begin{algorithm}
\caption{Approximate DP AM Protocol for Tagged Sample Retrieval}
\label{alg:label-invariant}
\textbf{Require:} \highlight{Privacy parameters $\eps>0$ and $\delta \in (0, 1)$}, proximity parameter $\sigma \in (0, 1)$, sample size $s\in \tO(N^{2/3}/\eps^{2/3}) \cdot \mathrm{poly}(\log(1/\delta)/\sigma)$, and sample access to a distribution $D$ over $[N]$ with $\max_{x \in [N]} D(x) \le p_{\mathrm{max}}$, where we define \highlight{$p_{\mathrm{max}} = \log(1/\delta)/\eps s$} and $p_{\mathrm{min}} = \sigma/N$. We also define $\tau = \sigma^3$ and  $B = \log(p_{\mathrm{max}}/p_{\mathrm{min}})/\tau$ and $b_{\mathrm{min}} = (\sigma/2B)s$ and \highlight{$\Delta = 3\log(s) + 5\log(1/\delta)/\eps$}.
\begin{enumerate}
    \item \label{step:seed-col} $V$ draws uniform samples $S_1, \ldots, S_s \iid [N]$. $V$ rejects if any $x \in [N]$ appears over $3$ times among these samples. Otherwise, $V$ sends these samples to $P$.
    \item \label{step:prover} $P$ sends back tags $\pi_1, \ldots, \pi_s$ to $V$, where
    \[
        \pi_i = \begin{cases}
            D(S_i) &\text{if }D(S_i) \ge p_{\mathrm{min}},\\
            0 &\text{if }D(S_i) < p_{\mathrm{min}}.
        \end{cases}
    \]
    \item \label{step:ptr} $V$ draws two fresh samples of size $s$ from $D$, denoted $T_1, \ldots, T_s \iid D$ and $T_1', \ldots, T_s' \iid D$. \highlight{$V$ calculates $m(T) = \max_{k \in [s]}\lvert\{r \in [s] \mid S_k = T_r\}\rvert$ and $m(T')$ similarly. $V$ rejects if} 
    \[
        \highlight{\Lap\Bigl(\max\bigl\{m(T), m(T')\bigr\}, \frac{1}{\eps}\Bigr) > \Delta - \frac{\log(1/\delta)}{\eps}.}
    \]
    \item \label{step:data-col} $V$ calculates, for each integer $j$ (possibly negative) such that $p_{\mathrm{min}} \le e^{j\tau}/N \le p_{\mathrm{max}}$, the bin
    \[
        S^j = \biggl\{i \in [s] \,\bigg|\, \pi_i \in \Bigl[\frac{e^{j\tau}}{N}, \frac{e^{(j+1)\tau}}{N}\Bigr)\biggr\}.
    \]
    $V$ calculates the number of two- and three-way collisions among $S, T, T'$ for each index $j$:
    \begin{align*}
        \pair_j &= \Bigl\lvert \bigl\{(k,r) \in [s]^2 \,\big|\, k \in S^j \text{ and }S_k = T_r\bigr\}\Bigr\rvert,\\
        \triple_j &= \Bigl\lvert\bigl\{(k,r,r') \in [s]^3 \,\big|\, k \in S^j\text{ and } S_k = T_r=T'_{r'}\bigr\}\Bigr\rvert.
    \end{align*}
    $V$ rejects if some bin with $\abs{S^j} \ge b_{\mathrm{min}}$ fails either of the following checks:
        \begin{align*}
            \highlight{\Lap\Bigl(\pair_j, \frac{3}{\eps}\Bigr)} &\in (1 \pm O(\tau)) \cdot s \abs{S^j} \frac{e^{j\tau}}{N}, \\
            \highlight{\Lap\Bigl(\triple_j, \frac{3\Delta}{\eps}\Bigr)} &\in (1 \pm O(\tau)) \cdot s^2\abs{S^j} \Bigl(\frac{e^{j\tau}}{N}\Bigr)^2.
        \end{align*}
    \item \label{step:light-col} $V$ calculates $S^{-\infty} = \{i \in [s] \,|\,\pi_i < p_{\mathrm{min}}\}$ and the corresponding two-way collision count:
    \[
        \pair_{-\infty} = \Bigabs{\bigl\{(k,r) \in [s]^2\,\big|\,k \in S^{-\infty} \text{ and } S_k = T_r\bigr\}}.
    \]
    $V$ rejects if the following check fails:
    \[
        \highlight{\Lap\Bigl(\pair_{-\infty}, \frac{3}{\eps}\Bigr) \le 2p_{\mathrm{min}} \cdot s^2.}
    \]
    
    \item $V$ accepts and outputs the tagged samples $(S_1, \pi_1), \ldots, (S_s, \pi_s)$.
\end{enumerate}
\end{algorithm}

To better understand the difference between the non-private protocol in \cite{herman2024public} and our new protocol, which is described in \Cref{alg:label-invariant}, let $V$ denote the verifier and $P$ denote the prover. As before, the key quantities are $\pair_j$ and $\triple_j$, which respectively count the number of two- and three-way collisions among elements sampled by the verifier (some from the uniform distribution over $[N]$, and some from $D$). We note that the sensitivity of $\triple_j$ (ie, the sensitivity of the number of 3-way collisions), is the maximum number of pairs $(S_k,T_r)$ where $S_k = T_r = x$ for some $x \in [N]$ where $S$ is the sample from uniform, $T$ is the sample from $D$, and $S_k \in S$ and $T_r \in T$. In general, this quantity, which is the maximum  number of 2-way collisions for any fixed value in $[N]$, could be as large as $s^2$ if both $S$ and $T$ are $s$ occurrences of a single $x \in [N]$.  Instead, following~\cite{dwork2009differential}, we  ``propose'' an upper bound on this sensitivity and test, in a differentially private fashion, that the proposed bound indeed exceeds the actual sensitivity of the number of 3-way collisions for the sample $S$.  This is known as the {\em local} sensitivity. The values are calibrated such that with probability $1-\delta$, if the local sensitivity is greater than what we allow, we will reject, and otherwise we will continue with the protocol, calibrating the privacy noise to the proposed sensitivity when deciding whether to reject based on the two-way and three-way collision tests.

Notably, \Cref{alg:non-private-label-invariant} and its analysis assume that the distribution $D$ is sufficiently \emph{dense}, or equivalently, has sufficiently high \emph{min-entropy}, meaning that $\max_{x \in [N]} D(x) \le 1/s$, where $s$ is the sample size required for the protocol. The justification for this assumption is that domain elements with mass exceeding $1/s$ can effectively filtered out in a preprocessing step using an additional $O(s)$ samples from the distribution. In the private case, we can similarly filter out elements with mass exceeding $\log(1/\delta)/\eps s$ using just $O(s)$ samples using an $(O(\eps), O(\delta))$-DP preprocessing step (e.g. using Theorem 7.1 of the survey \cite{vadhan2017dp}). Therefore, in \Cref{alg:label-invariant}, we assume that $D$ satisfies $\max_{x \in [N]} D(x) \le \log(1/\delta)/\eps s$, which is a weaker assumption than before.

Another difference between \Cref{alg:non-private-label-invariant} and \Cref{alg:label-invariant} is in the threshold used in the penultimate step, regarding $\pair_{-\infty}$. For technical reasons which will become clear over the course of the proof, we replace the factor of $\abs{S^{-\infty}}$ with the (larger) factor of $s$.

We now turn our attention to the privacy and utility analysis of our algorithm:

\begin{theorem}
\label{thm:label-invariant-privacy}
    \Cref{alg:label-invariant} satisfies $(O(\eps), O(\delta))$-DP for any choice of the algorithm's parameters.
\end{theorem}

\begin{proof}
    At a high level, \Cref{alg:label-invariant} is an instance of the \emph{propose-test-release (PTR)} framework in differentially private algorithm design \cite{dwork2009differential}. In our case, the statistics of interest are the two- and three-way collision counts $\pair_j$ and $\triple_j$ for various bin indices $j$, and we use PTR to control the local sensitivity of these quantities. In more detail, we first observe that the (global) sensitivity of $\pair_j$, which counts pairwise collisions between $S$ and $T$, is $3$. This is because Step~\ref{step:seed-col} of the algorithm ensures that no element among $S_1, \ldots, S_s$ appears more than $3$ times. Similarly, under the additional (data-dependent) assumption that no element in $T$ or $T'$ appears more than $\Delta$ times, the (local) sensitivity of $\triple_j$, which counts three-way collisions among $S$, $T$, and $T'$, is $3\Delta$. For this reason, in Step~\ref{step:ptr}, we propose a bound on $\max\{m(T), m(T')\}$, which is the number of relevant duplicate elements in either $T$ or $T'$. We observe that $m(T)$ and $m(T')$ both have sensitivity $1$, owing to the fact that replacing a single element in either $T$ or $T'$ can change either count by at most $1$. Therefore, testing this bound with the classic Laplace mechanism satisfies $(\eps, 0)$-DP.

    If the PTR test in Step~\ref{step:ptr} passes, then we know that the local sensitivities of $\pair_j$ and $\triple_j$ are bounded by $3$ and $3\Delta$, respectively, except in the rare (i.e. probability $\delta$) event that the Laplace noise in Step~\ref{step:ptr} had magnitude exceeding $\log(1/\delta)/\eps$. Under this local sensitivity bound, the additional instances of the Laplace mechanism that appear in Steps~\ref{step:data-col} and \ref{step:light-col} also satisfy $(\eps, 0)$-DP. We also know that each data point among $T$ and $T'$ affects $\pair_j$ and $\triple_j$ for at most one value of $j \in \Z \cup \{-\infty\}$. Therefore, in total, the composite algorithm satisfies $(O(\eps), O(\delta))$-DP.
\end{proof}

\begin{theorem}
\label{thm:label-invariant-utility}
    Given privacy parameters $\eps,\delta > 0$ and a proximity parameter $\sigma > 0$ that are not too small (i.e. satisfying $\log(1/\delta)/\eps\sigma \le N^{o(1)}$), along with access to $s = \tO(N^{2/3}/\eps^{2/3}) \cdot \mathrm{poly}(\log(1/\delta)/\sigma)$ samples from $D$, \Cref{alg:label-invariant} satisfies the following completeness and soundness conditions:
    \begin{itemize}
        \item \textbf{Completeness:} If the prover is honest, then with probability at least $3/4$, the verifier accepts and outputs tagged samples $(S_1, \pi_1), \ldots, (S_s, \pi_s)$ satisfying
        \[
            \frac{1}{s}\sum_{i \in [s] : \pi_i \ge \frac{\sigma}{N}} \biggl(1 - \min\Bigl\{\frac{\pi_i}{D(S_i)}, \frac{D(S_i)}{\pi_i}\Bigr\}\biggr) \le O(\sigma^3).
        \]
        and
        \[
            \frac{1}{s}\sum_{i\in[s]:\pi_i<\frac{\sigma}{N}} D(S_i) \le O\Bigl(\frac{\sigma}{N}\Bigr).
        \]
        \item \textbf{Soundness:} If the prover is dishonest, then with probability at least $3/4$, the verifier either rejects or outputs tagged samples $(S_1, \pi_1), \ldots, (S_s, \pi_s)$ satisfying
        \[
            \frac{1}{s}\sum_{i \in [s] : \pi_i \ge \frac{\sigma}{N}} \biggl(1 - \min\Bigl\{\frac{\pi_i}{D(S_i)}, \frac{D(S_i)}{\pi_i}\Bigr\}\biggr) \le O(\sigma).
        \]
        and
        \[
            \frac{1}{s}\sum_{i\in[s]:\pi_i<\frac{\sigma}{N}} D(S_i) \le O\Bigl(\frac{\sigma}{N}\Bigr).
        \]
    \end{itemize}
    (Note that the right-hand side of the first inequality in the completeness guarantee is $O(\sigma^3)$, but the corresponding quantity in the soundness guarantee is $O(\sigma)$. The sets of inequalities are otherwise identical up to suppressed constant factors.)
\end{theorem}

\begin{proof}
    The first proof ingredient that we need is a straightforward analysis of the mean and variance of the key statistics $\pair_j$ and $\triple_j$ over the randomness in the samples $T_i$ and $T'_i$. For this step, we do not consider the randomness over the data-independent samples $S_1, \ldots, S_s$, which we instead view as any fixed sequence of elements in $[N]$ such that no element appears more than $3$ times.
    A calculation of the mean and variance of $\pair_j$ and $\triple_j$ appeared in Appendix A of \cite{herman2024public}, but that analysis assumed that $p_{\mathrm{max}} \le 1/s$, so we cannot use it as a black box. In our private case, we only have the looser bound $p_{\mathrm{max}} \le \log(1/\delta)/\eps s$, and we take this into account by explicitly capturing the dependence on $p_{\mathrm{max}}$.

    \begin{claim}
    \label{thm:collision-variance}
        Define $\pair_j$, $\triple_j$, etc. as in the algorithm. Then,
        \begin{align*}
            \E[\pair_j] &= s\sum_{k \in S^j} D(S_k),\\
            \E[\triple_j] &= s^2\sum_{k \in S^j} D(S_k)^2.
        \end{align*}
        Also,
        \begin{align*}
            \Var[\pair_j] &\le 3\E[\pair_j],\\
            \Var[\triple_j] &\le 3(1+2p_{\mathrm{max}}s) \cdot \E[\triple_j].
        \end{align*}
    \end{claim}
    The proof of \Cref{thm:collision-variance} is quite standard, albeit slightly tedious. We simply write $\pair_j$ and $\triple_j$ as sums of indicator random variables and then apply properties like linearity of expectation. For this reason, we defer the calculation to the end of this section, focusing first on how to use \Cref{thm:collision-variance} to establish completeness and soundness.

    \paragraph{Proof of Completeness}

        First, we observe that regardless of whether the prover is honest or dishonest, the probability of rejection in Step~\ref{step:seed-col} is vanishingly small. Indeed, under the assumption in the theorem statement that $\eps,\delta,\sigma$ are not too miniscule, we have $s \le N^{2/3 + o(1)}$. Therefore, the probability of a four-way collision among the $S_i$ elements is at most $s^4N^{-3} \le N^{-1/3 + o(1)}$, which decays to $0$ as $N$ grows.

        Next, we analyze the probability of rejection in Step~\ref{step:ptr}. For this, we observe that the expected number of times any particular element appears among $T_1, \ldots, T_s$ is at most $p_{\mathrm{max}}s$. By standard concentration inequalities (e.g. the multiplicative Chernoff bound in \Cref{thm:concentration} for each $k \in [s]$ followed by a union bound over $k$), we have that $m(T) \le 2p_{\mathrm{max}}s + 2\log(s/\beta)$, except with failure probability $\beta$. The same is true of $m(T')$. For our choice of $p_{\mathrm{max}}$ and $\Delta$, it follows that even after adding Laplace noise, the verifier is very unlikely to reject in Step~\ref{step:ptr}, doing so with at most a small constant probability.

        Next, suppose the prover is honest, so that the tags $\pi_i = D(S_i)$ sent back to the verifier are legitimate. Under this assumption, we now analyze the probability of rejection in Step~\ref{step:data-col}. We start by analyzing $\pair_j$. By \Cref{thm:collision-variance}, for each bin $S^j$ of size at least $b_{\mathrm{min}} = (\sigma/2B)s$ with $e^{j\tau}/N \ge p_{\mathrm{min}} = \sigma/N$, the collision count $\pair_j$ has expectation
        \begin{equation*}
            \E[\pair_j] = s\sum_{i \in S^j} D(S_i) \ge s\abs{S^j}\frac{e^{j\tau}}{N} \ge \frac{s^2}{N} \cdot \frac{\sigma^2}{2B}. \tag{$*$}
        \end{equation*}
        Here, we have implicitly used the fact that the tags are legitimate in order to argue that $D(S_i) = \pi_i \ge e^{j\tau}/N$ for each index $i \in S^j$. Using the bound on $\Var[\pair_j]$ in \Cref{thm:collision-variance}, Chebyshev's inequality implies that
        \[
            \Pr\biggl[\pair_j \notin \Bigl(1 \pm O(\tau)\Bigr) \E[\pair_j] \biggr] \lesssim \frac{\Var[\pair_j]}{(\tau \E[\pair_j])^2} \lesssim \frac{1}{\tau^2 \E[\pair_j]} \lesssim \frac{N}{s^2} \cdot \frac{B}{\sigma^2\tau^2},
        \]
        
        which is at most a small constant for our chosen sample size $s$. Moreover, for our choice of $s$, we can ensure that for any particular small constant $c > 0$, the noise magnitude $3/\eps$ is less than $c\tau$ times our lower bound on $\E[\pair_j]$ in inequality $(*)$, namely $(s^2/N)(\sigma^2/2B) \le s\abs{S^j}(e^{j\tau}/N)$. Therefore, except with a small constant probability over the randomness of $T$, $T'$, and the various Laplace mechanisms, the following checks in Step~\ref{step:data-col} pass for all bins $S^j$ under consideration:
        \[
            \Lap\Bigl(\pair_j, \frac{3}{\eps}\Bigr) \in (1 \pm O(\tau))s\abs{S^j}\frac{e^{j\tau}}{N}.
        \]
        Next, for the three-way collision counts $\triple_j$, \Cref{thm:collision-variance} implies
        \begin{equation*}
            \E[\triple_j] = s^2\sum_{i \in S^j} D(S_i)^2 \ge s^2\abs{S^j}\Bigl(\frac{e^{j\tau}}{N}\Bigr)^2 \ge \frac{s^3}{N^2} \cdot \frac{\sigma^3}{2B}. \tag{$**$}
        \end{equation*}
        Using the formula for $\Var[\triple_j]$ in \Cref{thm:collision-variance}, Chebyshev's inequality implies that
        \[
            \Pr\biggl[\triple_j \notin \Bigl(1 \pm O(\tau)\Bigr)\E[\triple_j]\biggr] \lesssim \frac{\Var[\triple_j]}{(\tau\E[\triple_j])^2} \lesssim \frac{1+2p_{\mathrm{max}}s}{\tau^2\E[\triple_j]} \lesssim \frac{N^2}{s^3}  \cdot \frac{\log(1/\delta)B}{\eps\sigma^3\tau^2}.
        \]
        
        This, too, is at most a small constant for our choice of $s$. Again, our choice of $s$ ensures that for any particular small constant $c > 0$, the noise magnitude $3\Delta/\eps$ is less than $c\tau$ times our lower bound on $\E[\triple_j]$ in inequality $(**)$, namely $(s^3/N^2)(\sigma^3/2B) \le s^2\abs{S^j}(e^{j\tau}/N)^2$. In slightly more detail---since this is the most statistically expensive step of the proof---the relationship we wish to establish is:
        \[
            \frac{\Delta}{\eps} \asymp \frac{\log(s)}{\eps} + \frac{\log(1/\delta)}{\eps^2} \lesssim \frac{s^3}{N^2} \cdot \frac{\sigma^3}{2B}.
        \]
        For $s = \tO(N^{2/3}/\eps^{2/3}) \cdot \mathrm{poly}(\log(1/\delta)/\sigma)$, the asymptotic inequality above indeed holds, where we draw particular attention to the dependence on the parameters $N$ and $\eps$. Consequently, except with a small constant probability, the following checks pass for all bins $S^j$ under consideration:
        \[
            \Lap\Bigl(\triple_j, \frac{3\Delta}{\eps}\Bigr) \in (1 \pm O(\tau))s^2\abs{S^j}\Bigl(\frac{e^{j\tau}}{N}\Bigr)^2.
        \]

        Finally, we analyze the probability of rejection in Step~\ref{step:light-col}. For this, observe that when the prover is honest, the two-way collision count $\pair_{-\infty}$ satisfies
        \[
            \E[\pair_{-\infty}] = s \sum_{i \in S^{-\infty}} D(S_i) \le s\abs{S^{-\infty}}p_{\mathrm{min}} \le p_{\mathrm{min}} \cdot s^2.
        \]
        Here, we have used the fact that $\abs{S^{-\infty}} \le s$. Also, for our choice of $s$, the noise magnitude $3/\eps$ is at most an arbitrarily small constant times $(s^2/N) \cdot \sigma = p_{\mathrm{min}}\cdot s^2$. Consequently, except with small constant probability,
        \[
            \Lap\Bigl(\pair_{-\infty}, \frac{3}{\eps}\Bigr) \le 2p_{\mathrm{min}}\cdot s^2.
        \]
        Accumulating all the small constant failure probabilities discussed so far, we conclude that when the prover is honest, the verifier accepts with probability at least $3/4$. Moreover, the error of these tags $\pi_i$ is $0$ if exact, or proportional to $\tau = O(\sigma^3)$ if reported with $\tau$ precision.

    \paragraph{Proof of Soundness}

        As in the proof of completeness, we first observe that the probability of rejection in Steps~\ref{step:seed-col} or \ref{step:ptr} is vanishingly small, regardless of whether the prover is honest or dishonest. With this in mind, we shall adopt the following two-step approach to proving soundness:
        \begin{enumerate}
            \item In the first step, we will argue that the only way for a dishonest prover to pass the verifier's checks is to send tags $\pi_i$ such that the true probability masses in each sufficiently heavy bin $S^j$ have roughly the correct first two moments. Here, ``correct'' means matching that of an honest prover.
            
            Although this step of the proof is similar to one carried out in the non-private case, a crucial difference is that we must establish this implication even in the presence of the additive Laplace noise that we introduced for differential privacy.
            
            \item In the second step, we apply an argument from \cite{herman2024public}, building on ideas from the earlier work of \cite{batu2017uniformity}, to argue that having these accurate bin-wise moments is all we need. Since this step does not involve any changes for differential privacy, we invoke it as a black box.
        \end{enumerate}
    For the first step, consider any bin $S^j$ which is sufficiently heavy, meaning that $\abs{S^j} \ge b_{\mathrm{min}} = (\sigma/2B)s$. Our ultimate goal in this step is to establish that
    \begin{align}
        \E_{i \sim S^j} [D(S_i)] &\in \bigl(1 \pm O(\tau)\bigr) \frac{e^{j\tau}}{N}, \label{eq:balance-1}\\
        \E_{i \sim S^j} [D(S_i)^2] &\in \bigl(1 \pm O(\tau)\bigr)\Bigl(\frac{e^{j\tau}}{N}\Bigr)^2. \label{eq:balance-2}
    \end{align}
    To this end, we observe that these two quantities roughly correspond to $\pair_j$ and $\triple_j$, a relationship that is made precise by \Cref{thm:collision-variance}:
    \begin{align*}
        \E_{i \sim S^j} [D(S_i)] &= \frac{1}{\abs{S^j}s} \E[\pair_j],\\
        \E_{i \sim S^j} [D(S_i)^2] &= \frac{1}{\abs{S^j}s^2} \E[\triple_j].
    \end{align*}
    We emphasize that the expectations on the left-hand side are computed over the random index $i \sim S^j$, whereas the expectations on the right-hand side are computed over the random samples in $T$ and $T'$. In order to analyze $\E[\pair_j]$ and $\E[\triple_j]$, we apply Markov's inequality, which tells us that with probability at least $1 - O(\beta)$ over $T$ and $T'$,
    \begin{align*}
        \pair_j &\le \frac{1}{\beta}\E[\pair_j],\\
        \triple_j &\le \frac{1}{\beta}\E[\triple_j].
    \end{align*}
    If we also know that the verifier's checks on $\pair_j$ and $\triple_j$ in Step~\ref{step:data-col} pass, then we similarly have, with at least $1 - O(\beta)$ probability over the noise added for DP, that
    \begin{align*}
        \pair_j &\in (1 \pm O(\tau)) s \abs{S^j}\Bigl(\frac{e^{j\tau}}{N}\Bigr) \pm O{\biggl(\frac{\log(1/\beta)}{\eps}\biggr)}, \\
        \triple_j &\in (1 \pm O(\tau)) s^2 \abs{S^j}\Bigl(\frac{e^{j\tau}}{N}\Bigr)^2 \pm O{\biggl(\frac{\Delta\log(1/\beta)}{\eps}\biggr)}.
    \end{align*}
    Next, we substitute our bounds on $\abs{S^j} \ge b_{\mathrm{min}}$ and $e^{j\tau}/N \ge p_{\mathrm{min}}$ (and $\tau < 1/2$), as well as the previously mentioned relationship between $\pair_j$ and $\triple_j$ via Markov's inequality. Doing so simplifies the above displays substantially. Indeed, letting $\beta = O(1/B)$ where $B$ is the number of bins and taking a union bound over bins, we see that except with a small constant probability, for all bins of size $\abs{S^j} \ge b_{\mathrm{min}}$, we have
    \begin{align*}
        \E[\pair_j] &\ge \frac{s^2}{N} \cdot \mathrm{poly}\Bigl(\frac{1}{B}\Bigr) - \frac{1}{\eps} \cdot O(\mathrm{log}(B)),\\
        \E[\triple_j] &\ge \frac{s^3}{N^2} \cdot \mathrm{poly}\Bigl(\frac{1}{B}\Bigr) - \frac{\Delta}{\eps}\cdot O(\mathrm{log}(B)).
    \end{align*}

    To better understand these formulas, recall that $B = \mathrm{poly}(\log(N\log(1/\delta)/\eps s)/\sigma)$ and also that $\Delta = O(\log(s) + \log(1/\delta)/\eps)$. In particular, for our choice of $s$, both right-hand side quantities are positive and bounded away from $0$. Next, we combine these lower bounds on $\E[\pair_j]$ and $\E[\triple_j]$ with the upper bounds on $\Var[\pair_j]$ and $\Var[\triple_j]$ from \Cref{thm:collision-variance} in order to apply Chebyshev's inequality  as we did in the proof of completeness. Doing so, and substituting our choice of $s$, we see that except with a small constant probability, for all bins of size $\abs{S^j} \ge b_{\mathrm{min}}$, we have
    \begin{align*}
        \pair_j &\in \bigl(1 \pm O(\tau)\bigr)\E[\pair_j],\\
        \triple_j &\in \bigl(1 \pm O(\tau)\bigr) \E[\triple_j].
    \end{align*}
    At this point, we have related $\pair_j$ to its mean $\E[\pair_j]$ as well as to the quantity $s\abs{S^j}(e^{j\tau}/N)$. Similarly, we have related $\triple_j$ to both its mean $\E[\triple_j]$ and the quantity $s^2\abs{S^j}(e^{j\tau}/N)^2$. Combining these relationships yields (again, for our particular choice of $s$, except with a small constant failure probability, for all large enough bins $\abs{S^j} \ge b_{\mathrm{min}}$),
    \begin{align*}
        \E[\pair_j] &\in \bigl(1 \pm O(\tau)\bigr) s\abs{S^j}\frac{e^{j\tau}}{N},\\
        \E[\triple_j] &\in \bigl(1 \pm O(\tau)\bigr) s^2\abs{S^j}\Bigl(\frac{e^{j\tau}}{N}\Bigr)^2.
    \end{align*}
    This concludes the proof of the first step in our two-step proof of soundness. For the second step, we invoke as a black box the following claim, which follows from combining the proofs of several lemmas in Section 4.2 of \cite{herman2024public}:
    \begin{claim}[\cite{herman2024public}]
        If $\tau = O(\sigma^3)$ and conditions \eqref{eq:balance-1}, \eqref{eq:balance-2} hold for all bins of size $\abs{S^j} \ge b_{\mathrm{min}}$, then
        \[
            \frac{1}{s}\sum_{i \in [s] : \pi_i \ge \frac{\sigma}{N}} \biggl(1 - \min\Bigl\{\frac{\pi_i}{D(S_i)}, \frac{D(S_i)}{\pi_i}\Bigr\}\biggr) \le O(\sigma).
        \]
    \end{claim}
    The only bin we have not yet addressed is $S^{-\infty} = \{i \in [s] : \pi_i < \sigma/N\}$. For this bin, we first observe that its mean satisfies
    \[
        \E[\pair_{-\infty}] = s\sum_{i \in [s]:\pi_i<\frac{\sigma}{N}} D(S_i)
    \]
    By standard Chernoff/Hoeffding bounds and a union bound over $i \in [s]$, with all but a small constant probability, we have
    \[
        \pair_{-\infty} \ge s\sum_{i \in [s]:\pi_i<\frac{\sigma}{N}} D(S_i) - \tO(\sqrt{s}).
    \]
    If the verifier's check in Step~\ref{step:light-col} passes, then with all but a small constant probability, we have
    \[
        \pair_{-\infty} \lesssim p_{\mathrm{min}} \cdot s^2 + \frac{1}{\eps}.
    \]
    Combining the upper and lower bounds on $\pair_{-\infty}$, rearranging terms, and substituting our choice of the sample $s$, we conclude that
    \[
        \frac{1}{s} \sum_{i \in [s]: \pi_i < \frac{\sigma}{N}} D(S_i) \lesssim p_{\mathrm{min}} = \frac{\sigma}{N}.
    \]
    This concludes the proof of soundness.
    
    \paragraph{Deferred Proof of \Cref{thm:collision-variance}}
    
    Fix a bin $j$. Let $\tilde{S}^j$ be a deduplicated version of $S^j$ containing one representative index for each distinct sample value in $S^j$. Let $m_k$ be the multiplicity of $S_k$ within bin $S^j$, and let $p_k$ be its true probability under $D$. Let $\pair_{jk}$ and $\triple_{jk}$ count the number of two- and three-way collisions with a specific element $S_k$, as opposed to all $S_k$ with $k \in S^j$. We can decompose
    \begin{align*}
        \pair_j &= \sum_{k \in \tS^j} m_k\cdot \pair_{jk},\\
        \triple_j &= \sum_{k \in \tS^j} m_k\cdot \triple_{jk}.
    \end{align*}
    By linearity of expectation and \Cref{thm:binom} below with $X = \pair_{jk}$ and $XY = \triple_{jk}$ and $p = p_k$,
    \begin{align*}
        \E[\pair_j] &= \sum_{k \in \tS^j} m_{k}\cdot sp_k,\\
        \E[\triple_j] &= \sum_{k \in \tS^j} m_k \cdot s^2p_k^2.
    \end{align*}
    Since $\pair_{jk}$ and $\pair_{jk'}$ are negatively correlated for distinct $k,k' \in \tS^j$ (as are $\triple_{jk}$ and $\triple_{jk'}$),
    \begin{align*}
        \Var[\pair_j] &\le \sum_{k \in \tS^j} \Var[m_k\cdot \pair_{jk}],\\
        \Var[\triple_j] &\le \sum_{k \in \tS^j} \Var[m_k\cdot \triple_{jk}].
    \end{align*}
    Applying \Cref{thm:binom} again yields
    \begin{align*}
        \Var[\pair_j] &\le \sum_{k \in \tS^j} m_k^2 \cdot sp_k(1-p_k),\\
        \Var[\triple_j] &\le \sum_{k \in \tS^j} m_k^2 \cdot s^2p_k^2(1-p_k)(1-p_k+2sp_k).
    \end{align*}
    Combining our mean and variance formulas concludes the proof.
    
    \begin{lemma}[Binomial Mean and Variance]
    \label{thm:binom}
        If $X,Y \iid \Bin(s,p)$, then
        \begin{align*}
            \E[X] &= sp, \\
            \Var[X] &= sp(1-p), \\
            \E[XY] &= s^2p^2, \\
            \Var[XY] &= s^2p^2(1-p)(1-p+2sp).
        \end{align*}
    \end{lemma}
    \begin{lemma}[Chernoff Variant]
    \label{thm:concentration}
        If $X$ is a sum of independent $[0, 1]$-valued random variables, then
        \[
            \Pr\bigl[X > 2\E[X] + 2\log(1/\beta)\bigr] \le \beta.
        \]
    \end{lemma}
    \begin{proof}
        Let $\mu = \E[X]$. By the standard multiplicative Chernoff bound, for all $\delta > 0$,
        \[
            \Pr[X \ge (1 + \delta)\mu] \le e^{-\delta^2\mu/(2+\delta)}.
        \]
        Since $2 + \delta \le 2\max(2,\delta)$, the right-hand side is at most $\max(e^{-\delta^2\mu/4}, e^{-\delta\mu/2})$. This quantity is equal to $\beta$ when we set
        \(
            \delta = \max(2\log(1/\beta)/\mu,  2\sqrt{\log(1/\beta)/\mu}).
        \)
        For this value of $\delta$, we have
        \begin{align*}
            (1 + \delta)\mu &= \mu + \max\Bigl(2\log(1/\beta), 2\sqrt{\mu\log(1/\beta)}\Bigr)\\
            &\le \mu + \max\bigl(2\log(1/\beta), \mu + \log(1/\beta)\bigr) &\text{(AM-GM Inequality)}\\
            &= \max\bigl(\mu + 2\log(1/\beta), 2\mu + \log(1/\beta)\bigr)\\
            &\le 2\mu + 2\log(1/\beta). \qedhere
        \end{align*}
    \end{proof}
\end{proof}

\section{Future Work}
\begin{enumerate}
    \item Is $\tilde{O}(N^{2/3})$ the correct bound for public coin (non-DP) interactive proofs with label invariant properties? Can AM protocols instead match our current best-known lower bound of $\tilde{O}(N^{1/2})$ from \cite{chiesa_et_al:LIPIcs.ITCS.2018.53}?
    \item In the setting of distribution testing, is there an interesting notion of public coin protocols where the verifier can send messages that depend on their sample, with the restriction that all randomness used to compute the message is shared with the prover? The current notion of AM protocols was very natural in functional property testing as the verifier had no secret information, but this condition does not translate over to distribution testing.
    \item Is there a weaker adversary model that makes sense to study? For example, is there some real world adversary which may not have any information about your sample beyond knowing the distribution? Protecting against membership inference against such an adversary in communication plus DP on the output may be easier. We need to make such assumptions for soundness anyways so it may be reasonable to do so for privacy as well.
\end{enumerate}

\paragraph{AI Disclosure:} We used ChatGPT 5.2 Pro to assist in proving the lower bound in section 5.2. Specifically, the claim that when we have a distribution with approximately uniform marginals, the distance to the product distribution with the same marginals is at most $O\left(\sqrt{d}\right)$ times the distance to the nearest product distribution was proved with LLM assistance. The authors modified and verified the correctness of this proof.

Additionally, GPT 6.1 Sol was used in the proof in section 4. We came up with the core ideas for the reduction, while details for the acceptance probability and proofs were produced with LLM assistance. The authors also modified and verified the correctness of this proof.

\section{Acknowledgements}
We thank Elena Ghazi for spotting a mistake in our protocol for the reduction in section 4 in a previous version of the paper. 

\bibliographystyle{alpha}
\bibliography{ref}
\appendix
\section{Errata in Previous Version}
\label{sec:errata}
A previous version of this paper erroneously ignored an important detail in Goldreich's reduction, to wit, the fact that samples of $P$ are used to select locations $x$ of $Q$ at which to query the weights $Q(x)$.  There are no privacy implications when the verifier has $Q$.  However, in our interactive argument system in which the prover commits to $Q$ using a Merkle tree and then responds to queries $Q(x)$, revealing $x$'s sampled from $P$ destroys privacy.  We therefore withdraw Section 8 (Differentially Private Arguments of Proximity) of~\cite{oldversion}.

Our error also affected our claimed efficient upper bound for an AM protocol for testing that $P$ is a product distribution is a product distribution.  Computational efficiency is {\em not} a property of the algorithms in~\cite{aliakbarpour2017differentiallyprivateidentitycloseness}.  Our new algorithm in Section~\ref{sec:independence} closes this gap.  

We thank the anonymous reviewer for pointing out our mistake.

\section{Proof of Lemma 1 in \cite{Paninski}}
\label{app:Paninski}

Here, we adapt Paninski's proof and demonstrate that it is still true under the weaker assumption that $s < M$, in comparison to the assumption that $s = o(M)$ in \cite{Paninski}.

\begin{lemma} Suppose $2\le s<M$. For a sample of size $s$ drawn from any distribution $R$ on an alphabet of size $M$, there exists a universal constant $c_{\rm gap}>0$ such that
$$\mu_M-\mathbb E_R K\ge c_{\rm gap}\frac{s^2}{M}\TV(R,U_M)^2,$$
where $K$ is the number of elements that occur exactly once in the sample.
\end{lemma}

\begin{proof}
    Let $p_i$ be the mass that $R$ places on element $i$ and let
    $$f(p_i) = p_i \left[1-\left(\frac{M}{M-1}(1-p_i)\right)^{s-1}\right].$$

The function $f(x)$ then satisfies the following properties:
    \begin{enumerate}
        \item $f(0) = 0$
        \item $f(1/M) = 0$
        \item $f(x) < 0$ for $0 < x < 1/M$ 
        \item $f(x)$ is monotonically increasing for $x > 1/M$
        \item $f(x) \rightarrow x$ for sufficiently large $x$
    \end{enumerate}
    
    We are going to use these properties to lower bound the following term:

    $$\mu_M - \mathbb E_R K = s\left(\frac{M-1}{M}\right)^{s-1} \sum_{i=1}^M f(p_i).$$

To lower bound this sum, we first subtract the tangent to $f$ at $1/M$;
the linear terms cancel since $\sum_i(p_i-1/M)=0$.
Specifically, put $u=1/M$ and define
$$r(x)=f(x)-f'(u)(x-u),\qquad f'(u)=\frac{s-1}{M-1}\le1.$$
Since $f(u)=0$, we have $r(u)=r'(u)=0$. We claim a uniform lower bound that
$$r(x)\ge\frac{s-1}{2}(u-x)_+^2\qquad(0\le x\le1),$$
where $(a)_+=\max\{a,0\}$.

For $0\le x\le u$ and $s\ge3$,
$$f''(x)=\left(\frac{M}{M-1}\right)^{s-1}
(s-1)(1-x)^{s-3}(2-sx)
\ge(s-1)(1-u)^{-2}\ge s-1,$$
because $x\le u$ and $su\le1$. 

For $0\le x\le u$ and $s=2$, the same bound follows
from $f''(x)=2M/(M-1)$. 

Therefore, the integral remainder formula gives
$$r(x)=\int_x^u(u-t)f''(t)\,dt
\ge\frac{s-1}{2}(u-x)^2.$$

Next, consider $x\ge u$. On $[u,1/s]$, we have $f''(x)\ge0$,
so $f'(x)\ge f'(u)$.
For $1/s\le x\le1$,
$$f'(x)=1-\left(\frac{M}{M-1}\right)^{s-1}
(1-x)^{s-2}(1-sx)\ge1\ge f'(u),$$
with endpoint values understood by continuity. Consequently $r(x)\ge0$
for $x\ge u$, proving the claim.

Writing $d=\TV(R,U_M)=\sum_{i:p_i<u}(u-p_i)$, then Cauchy--Schwarz inequality gives
\begin{align*}
\sum_{i=1}^M f(p_i)
&\ge\sum_{i=1}^M r(p_i)\\
&\ge\frac{s-1}{2}\sum_{i:p_i<u}(u-p_i)^2
\ge\frac{s-1}{2M}d^2.
\end{align*}
Substituting this bound into the expression for $\mu_M-\mathbb E_R K$
above, and using $(1-1/M)^{s-1}\ge(1-1/M)^{M-1}\ge e^{-1}$ and
$s-1\ge s/2$, we obtain
$$\mu_M-\E_R K
\ge\frac{s(s-1)}{2eM}\TV(R,U_M)^2
\ge\frac{1}{4e}\frac{s^2}{M}\TV(R,U_M)^2.$$
This proves the lemma with $c_{\rm gap}=1/(4e)$.\qedhere
\end{proof}

\end{document}